\documentclass[twocolumn]{autart}    

\usepackage{array}
\usepackage{graphicx}          
\usepackage{makecell}
\usepackage{url}
\usepackage{amsmath,amssymb}
\usepackage{bm}
\usepackage{mathtools}
\usepackage{algorithm}
\usepackage{algorithmic}
\usepackage{enumitem}
\usepackage{xcolor}
\usepackage[normalem]{ulem}
\usepackage[T1]{fontenc}
\usepackage[utf8]{inputenc}

\newtheorem{theorem}{Theorem}
\newtheorem{assumption}{Assumption}
\newtheorem{remark}{Remark}
\newtheorem{lemma}{Lemma}

\newtheorem{proof}{Proof}

\newcommand{\cB}{\mathcal{B}}

\newcommand{\cH}{\mathcal{H}}

\newcommand{\bE}{\mathbf{E}}

\newcommand{\bF}{\mathbf{F}}
\newcommand{\bpsi}{\bm{\psi}}

\newcommand{\card}{\mathop{\rm card}\nolimits}

\newcommand{\Kop}{\mathcal{K}}

\newcommand\norm[1]{\left\lVert#1\right\rVert}
\newcommand{\rr}{{\mathbb R}}

\definecolor{revblue}{RGB}{0,90,180}
\definecolor{revred}{RGB}{200,30,45}
\definecolor{newrevgreen}{RGB}{0,128,0}
\usepackage{caption}

\newcommand{\revadd}[1]{#1}
\newcommand{\revdel}[1]{}

\renewcommand{\arraystretch}{1.2}
\AddToHook{cmd/normalsize/after}{%
    \setlength{\abovedisplayskip}{10pt plus 2pt minus 2pt}%
    \setlength{\belowdisplayskip}{10pt plus 2pt minus 2pt}%
}
\begin{document}

\begin{frontmatter}
\title{A Least-Squares Multi-Step Koopman Operator for\\ Model Predictive Control\thanksref{footnoteinfo}} 
\thanks[footnoteinfo]{This paper was not presented at any IFAC 
meeting.\\
$^*$Equal Contributions. $^{\dagger}$Corresponding Author: Liang Wu.
}

\author[MIT]{Wallace Gian Yion Tan$^*$}\ead{wtgy@mit.edu},               
\author[MIT]{Leqi Zhou$^*$}\ead{leqizhou@mit.edu}, 
\author[MIT]{Richard D. Braatz}\ead{braatz@mit.edu},  
\author[JHU]{ J\'an Drgo\v na}\ead{jdrgona1@jh.edu},
\author[MIT,JHU]{Liang Wu$^{*,\dagger}$}\ead{liangwu@mit.edu}
\address[MIT]{Massachusetts Institute of Technology, MA 02139, USA.} 
\address[JHU]{Johns Hopkins University, MD 21218, USA.}


\begin{abstract}                          
MPC is widely used in real-time applications, but practical implementations are typically restricted to convex QP formulations to ensure fast and time-certified execution. Koopman-based MPC enables QP-based control of nonlinear systems by lifting the dynamics to a higher-dimensional linear representation. However, existing approaches rely on single-step EDMD. Consequently, prediction errors may accumulate over long horizons when the EDMD operator is applied recursively. Moreover, the multi-step prediction loss is nonconvex with respect to the single-step EDMD operator, making long-horizon model identification particularly challenging. This paper proposes a multi-step EDMD framework that directly learns the condensed multi-step state control mapping required for Koopman-MPC, thereby bypassing explicit identification of the lifted system matrices and subsequent model condensation. The resulting identification problem admits a convex least-squares formulation. We further show that the problem decomposes across prediction horizons and state coordinates, enabling parallel computation that satisfies causality constraints and row-wise $\ell_1$-regularization for automatic dictionary pruning. A non-asymptotic finite-sample analysis demonstrates that, unlike single-step EDMD, the proposed method avoids error compounding and yields error bounds that depend only on the target multi-step mapping. Numerical examples validate improved long-horizon prediction accuracy and closed-loop performance.
\end{abstract}

\end{frontmatter}
\setlength{\parskip}{10pt plus 2pt minus 2pt}
\setlength{\abovedisplayskip}{10pt plus 2pt minus 2pt}
\setlength{\belowdisplayskip}{10pt plus 2pt minus 2pt}

\section{Introduction}
Model predictive control (MPC) is a model-based optimal control framework that has been widely adopted in manufacturing, energy systems, and robotics. To enable reliable real-time execution, MPC is commonly formulated as a convex quadratic program (QP), since QPs admit fast numerical solution methods \cite{ferreau2014qpoases,stellato2020osqp,wu2023simple} and allow execution-time-certified computation (ensuring that worst-case computation times satisfy real-time feedback constraints) \cite{wu2025direct,wu2025eiqp,wu2025time} on industrial embedded platforms. 

However, the standard QP formulation of MPC requires linear prediction models, affine constraints, and a convex quadratic objective, which hinders its direct applicability in nonlinear dynamical systems. A common workaround, known as model-based approximation, is to approximate the nonlinear dynamics via online linearization about the current state or a previous solution trajectory, resulting in a sequence of QPs \cite{berberich2022linear,gros2020linear}. 

Koopman-based lifting provides a data-driven alternative for constructing linear predictors of nonlinear systems. By learning linear predictors in a high-dimensional observable space, Koopman-based methods allow nonlinear dynamics to be represented by lifted linear models suitable for predictive control \cite{korda2018linear,zhang2022robust,williams2016extending,shi2022deep,mauroy2019koopman}. This richer parameterization explains the improved closed-loop performance often observed in Koopman-MPC \cite{korda2018linear,zhang2022robust}. In conventional Koopman-MPC, one first identifies a lifted linear state-space model of the linear form
\begin{equation}\label{eq: one step lifted model}
z_{k+1}=Az_k+Bu_k, \qquad x_k=Cz_k, \qquad z_k=\phi(x_k),
\end{equation}
where $\phi(\cdot)$ denotes the lifting map, $z_k$ is the lifted state, and $\{A,B,C\}$ are the identified lifted state-space matrices. Using \eqref{eq: one step lifted model} over a finite MPC horizon, the predicted lifted states can be algebraically eliminated by recursively substituting the lifted dynamics. In the MPC literature, this elimination step is commonly referred to as model condensation \cite{rawlings2020model,jerez2012sparse}. This yields a dense input-sequence QP whose decision variables are only the control inputs, and whose finite-horizon state-control prediction map takes the form
\begin{equation}\label{eq: state-control mapping}
(x_{1},\cdots,x_{H})
=
\mathbf{E}\phi(x(t))
+
\mathbf{F}(u_0,\cdots,u_{H-1}).
\end{equation}
\subsection{Related work}
The most widely used identification method in Koopman-based MPC is Extended Dynamic Mode Decomposition (EDMD) \cite{williams2015data}, which first specifies a dictionary of nonlinear lifting observables $\phi(\cdots)$ \emph{a priori} and then solves a convex least-squares problem to identify the linear state-space matrices $\{A,B,C\}$. The finite-horizon matrices $\bE$ and $\bF$ in \eqref{eq: state-control mapping} are then identified by recursive condensation through powers of $A$. EDMD offers scalable computation and algorithmic simplicity\revdel{ due to its reliance on least-squares estimation}, but single-step predictors tend to accumulate prediction error as the MPC horizon increases. This may lead to degradation of open-loop and closed-loop performance \cite{su1993neural}, especially when the spectral radius of the Koopman matrix $A$ exceeds unity\revdel{, in which case prediction errors are amplified} \cite{su1993neural,dahdah2022system,korda2020optimal}.
\revdel{Moreover, its least-squares structure enables analytical bounds on model approximation errors in Koopman-MPC, which provides a tractable foundation for theoretical analysis \cite{strasser2026overview,zhangzuazua2022,NuskePeitzPhilippSchallerWorthmann2023,philipp2024variancerepresentationsconvergencerates}.}

In deep-neural-network (DNN)-based Koopman approaches \cite{li2017extended,yeung2019learning}, the observables and the linear state-space matrices $\{A,B,C\}$ are jointly identified through a nonconvex learning problem. Within the DNN framework, the training objective can naturally be extended to a multi-step prediction loss \cite{xiao2022deep}. However, these approaches utilize a loss function that necessitates terms of the form $\{A^k,\cdots,A\}$, which breaks the convexity of the original least-squares problem and results in a nonconvex optimization formulation; see \cite{sayed2024novel,abtahi2025multi,di2024stable}. Minimizing this loss promotes high-fidelity multi-step prediction, which is essential in finite-horizon MPC. Although these methods improve multi-step accuracy, they still require recursive prediction and subsequent condensation into the form \eqref{eq: state-control mapping} for MPC applications. Another recent approach to improving multi-step prediction is \cite{sayed2024recursive}, which proposes a recursive least-squares method for multi-step prediction and incorporates DNNs for encoding and decoding from the lifted observable space. Notably, this approach utilizes iterative least-squares updates but still ultimately relies on identifying state-space matrices $\{A,B,C\}$ and subsequently condensing them for MPC, rather than directly identifying an explicit multi-step state-control mapping such as \eqref{eq: state-control mapping}.

Since MPC ultimately utilizes the multi-step state-control mapping \eqref{eq: state-control mapping}, \cite{de2024koopman} similarly adopts the viewpoint of directly learning the condensed multi-step matrices $\bE$ and $\bF$. The proposed framework in \cite{de2024koopman} provides recursive feasibility and stability guarantees, but its learning framework does not explicitly enforce the block lower-triangular structure of $\bF$ required for causality. In contrast, for a fixed set of observables, this work provides a horizon-wise parallel decomposition of the learning problem for $\bE$ and $\bF$ that guarantees causality of the resulting predictor.

\subsection{Contributions}
Motivated by this observation, this paper proposes a multi-step EDMD algorithm that directly learns the matrices $\{\mathbf{E},\mathbf{F}\}$ governing the multi-step state-control mapping, thereby bypassing the identification of $\{A,B,C\}$ and the subsequent construction of $\{\mathbf{E},\mathbf{F}\}$ used in conventional Koopman-MPC approaches. As a result, our algorithm admits a convex least-squares formulation. In addition, this paper shows that the proposed multi-step EDMD identification problem can be decomposed at the prediction-horizon and state-coordinate levels, enabling parallelized identification \revdel{ of $\{\mathbf{E},\mathbf{F}\}$ and facilitating the incorporation of row-wise $\ell_1$ regularization on $\mathbf{E}$ for subsequent dictionary pruning.}of $\{\mathbf{E},\mathbf{F}\}$. \revadd{The primary purposes of parallelization are twofold: 1) to impose the block lower-triangular structure of $\bF$, thereby ensuring causality so that controls do not affect state predictions at earlier horizons, and 2) to facilitate the incorporation of row-wise $\ell_1$ regularization on $\mathbf{E}$ for dictionary pruning.} Dictionary pruning mitigates the difficulty of dictionary selection by automatically removing functions that are irrelevant to the system dynamics. Consequently, we obtain a least-squares-based multi-step EDMD algorithm with parallelization and integrated dictionary pruning. Table \ref{tab:related-methods} summarizes the existing approaches in the literature and compares them with our framework.

Furthermore, we provide a non-asymptotic analysis of both single-step and multi-step EDMD under finite data regimes. Under a stationary-data assumption, we first derive an explicit high-probability error-propagation bound for recursively applied single-step EDMD, which reveals that the error grows exponentially in the horizon, where the growth rate is the spectral radius of the learned Koopman matrix. For multi-step EDMD, we derive a high-probability bound for our class of causally-consistent multi-step predictors, whose dependence on the horizon is polynomial, rather than exponential like in single-step EDMD. This distinction provides a principled explanation for the observation that our proposed multi-step EDMD approach improves long-horizon prediction performance.
\begin{table*}[t]
\centering
\scriptsize

\refstepcounter{table}
\label{tab:related-methods}

\begin{center}
\textbf{Table~\thetable. Comparison with related Koopman-based prediction methods for control.}
\end{center}

\setlength{\tabcolsep}{3pt}
\renewcommand{\arraystretch}{1.0}

\begin{tabular}{p{0.15\textwidth}
                p{0.15\textwidth}
                p{0.17\textwidth}
                p{0.19\textwidth}
                p{0.24\textwidth}}
\hline
\multicolumn{1}{>{\centering\arraybackslash}p{0.15\textwidth}}{\textbf{Method}}
&
\multicolumn{1}{>{\centering\arraybackslash}p{0.17\textwidth}}{\textbf{Learned object}}
&
\multicolumn{1}{>{\centering\arraybackslash}p{0.18\textwidth}}{\textbf{Optimization}}
&
\multicolumn{1}{>{\centering\arraybackslash}p{0.20\textwidth}}{\textbf{MPC condensed form}}
&
\multicolumn{1}{>{\centering\arraybackslash}p{0.24\textwidth}}{\textbf{Key distinction / tradeoff}}
\\
\hline

Single-step EDMD~\cite{korda2018linear}
&
Recursive lifted model $\{A,B,C\}$
&
Convex least squares
&
Requires recursive condensation through powers of $A$
&
Simple and scalable, but single-step errors can propagate over the horizon
\\

DNN Koopman methods~\cite{li2017extended,yeung2019learning,xiao2022deep}
&
Learned observables and recursive lifted model $\{A,B,C\}$
&
Nonconvex neural-network training
&
Usually requires recursive prediction and subsequent condensation
&
Highly expressive, but architecture-dependent and harder to analyze
\\

RLS multi-step Koopman~\cite{sayed2024recursive}
&
Recursive lifted model $\{A,B\}$ trained using multi-step rollouts
&
Iterative/RLS-style least-squares updates
&
Still requires recursive rollout or condensation from $\{A,B\}$
&
Improves multi-step rollout accuracy while retaining a recursive model; does not directly learn $\{\mathbf E,\mathbf F\}$
\\

Robust Koopman DPC~\cite{de2024koopman}
&
Condensed finite-horizon matrices $\{\mathbf E,\mathbf F\}$ and observables
&
Observable learning plus global least-squares predictor
&
Directly learned but $\mathbf F$ not causal
&
Provides stability/feasibility guarantees, but does not explicitly enforce block lower-triangular causality or horizon-wise parallel decomposition
\\

Proposed multi-step EDMD
&
Condensed finite-horizon matrices $(\mathbf E,\mathbf F)$
&
Convex least squares, decomposed by horizon and state coordinate
&
Directly learned with $\mathbf F$ causal
&
MPC-ready, parallelizable, supports pruning, and avoids recursive error amplification
\\

\hline
\end{tabular}
\end{table*}

\subsection{Notations}
Let $\mathbb{Z}^+$ denote the set of positive integers and $\mathbb{N} := \{0\} \cup\mathbb{Z}^+ $ denote the set of natural numbers. Given a matrix $P \in \mathbb{R}^{n \times n}$, $P \succeq 0 $ (resp., $P \succ 0$) denotes that $P$ is symmetric positive semi-definite (resp., symmetric positive definite). $\otimes$ denotes the standard Kronecker product. \revdel{Given vectors $v_1, v_2, \cdots{},v_n$, let $\operatorname{vec}[v_1, \cdots{},v_n ] := [v_1^\top, v_2^\top, \cdots{}, v_n^\top ]^\top$.} Let $\mathbf1^n \in \mathbb{R}^n$ denote the vector with ones in all components. Let $\norm{\cdot}_2$ and $\norm{\cdot}_1$ denote the standard Euclidean and $\ell^1$ norms of $v$, respectively. If $X$ and $Y$ are Banach or Hilbert spaces and $A : X \to Y$ is a linear operator, $\cB(X,Y)$ denotes the set of bounded linear operators between $X$ and $Y$, and for $A \in \cB(X,Y)$, $\norm{A}_{X\to Y}$ denotes the operator norm of $A$, which is given by
$$ \norm{A}_{X\to Y} := \sup_{x \in X, \norm{x}_X  = 1} \norm{Ax}_Y$$
If $X = Y$, denote $\cB(X,X)$ as $\cB(X)$ and $\norm{A}_{X}$ denotes the operator norm of $A$. If $A \in \mathbb{R}^{n \times m}$, denote $A^{\ddagger} \in \mathbb{R}^{m \times n} $ as the Moore-Penrose inverse of $A$. $a\gg b$ means $a\ge Cb$ for a sufficiently large universal constant $C>0$.
\vspace{-5pt}
\section{Problem Formulation and Preliminaries}\label{sec: prelim}
Consider the problem of regulating a nonlinear discrete-time system to the origin:
\begin{equation}\label{eqn_sys}
    x_{t+1}=f(x_t,u_t),
\end{equation}
where $x_t\in \mathbb{R}^{n_x}$ and \revadd{$u_t\in \mathcal{U} \subseteq  \mathbb{R}^{n_u}$ } denote the system state and control input at the sampling time $t\in\mathbb{Z}^+$\revadd{, and $\mathcal{U}$ denotes the space of all admissible control inputs}. 
MPC addresses this problem by solving a finite-time optimal control problem subject to the dynamics in \eqref{eqn_sys} and additional state and control input constraints. Assuming that a full measurement of the state $x(t)$ at the sampling time $t$ is available, the finite-horizon optimal control problem is shown in \eqref{eqn_nmpc}.
This optimization problem is solved at each sampling time $t$, where $x_{t+k}$ denotes the predicted state vector at time $t+k$ obtained by applying the control input sequence $u_t,\cdots,u_{t+k-1}$ to the system \eqref{eqn_sys} starting from the state $x(t)$, and $N$ denotes the prediction horizon. In \eqref{eqn_nmpc}, we further assume that $Q=Q^\top\succeq0,\ R=R^\top\succ0,\ P=P^\top\succeq 0$. The decision variables in the optimization problem \eqref{eqn_nmpc} are the control inputs $U:=\left\{u_{t},u_{t+1},\cdots{},u_{t+H-1} \right\}$. The corresponding state trajectory $X:=\left\{x_{t+1},\cdots{},x_{t+H} \right\}$ is implicitly defined by the nonlinear system dynamics \eqref{subeq:nmpc-dyn} and the control input sequence $U$. The nonlinear MPC problem in general is a nonlinear program due to the nonlinear dynamics arising from \eqref{subeq:nmpc-dyn}. Solving the nonlinear program at every sampling time can be computationally expensive due to the general intractability of such problems, especially for the MPC problem with long prediction horizons. To mitigate the computational burden of this problem, the Koopman operator lifts the nonlinear system \eqref{eqn_sys} into a higher-dimensional space of observables, where the system evolves with linearized dynamics. This transformation allows the reformulation of \eqref{eqn_nmpc} into a convex QP that can be solved using standard QP solvers such as interior point methods.  
\noindent\rule{\columnwidth}{0.4pt}
\textbf{Nonlinear MPC}
\begin{subequations}\label{eqn_nmpc}
\begin{align}
\min_{U} \ & x_{t+H}^\top P x_{t+H} +\sum_{k=0}^{H-1}\left[x_{t+k}^\top Q x_{t+k} + u_{t+k}^\top R u_{t+k} \right] \label{subeq:nmpc-cost} \\
\text{s.t.} \ & x_{t}=x(t), \\
& x_{t+k+1} = f(x_{t+k}, u_{t+k}), \ k=0,\dots,H-1  \ \label{subeq:nmpc-dyn} \\
&   x_{\min} \leq x_{t+k+1}  \leq x_{\max}, \ k=0,\dots,H-1 \label{subeq:mpc-state-constraints} \\
& u_{\min} \leq u_{t+k} \leq u_{\max}, \ k=0,\dots,H-1
\end{align}
\end{subequations}
\noindent\rule{\columnwidth}{0.4pt}

\subsection{Koopman operator and Extended Dynamic Mode Decomposition}\label{subsec:EDMD}
Koopman \cite{koopman1931hamiltonian} proposed an approach grounded in operator-theoretic methods to represent uncontrolled discrete-time nonlinear dynamical systems $ x_{t+1}=f(x_t)$ linearly. Specifically, Koopman demonstrated the existence of an infinite-dimensional linear operator $\mathcal{K}$, which governs the evolution of an infinite-dimensional Hilbert space of observables. Given an observable function $\psi : \mathbb{R}^{n_x} \to \mathbb{R}$, the Koopman operator acts on $\psi$, and $\mathcal{K}\psi$ is defined as
\begin{equation}\label{eqn_KoopmaM_Hef}
    \mathcal{K} \psi := \psi\circ f.
\end{equation}

Although the dynamics $f$ might be nonlinear, the Koopman operator $\Kop$ is linear on the space of observables. In practice, $\Kop$ is infinite-dimensional and must be approximated using a finite-dimensional dictionary of observables.

\subsubsection{Koopman Operator for Controlled Systems}
Although the Koopman operator was initially described for autonomous dynamical systems, numerous schemes (see \cite{williams2016extending,proctor2018generalizing,korda2018linear,strasser2026overview}) have been proposed to extend the application of the Koopman operator to controlled systems of the form \eqref{eqn_sys}. In this work, we adopt the Koopman Control Family (KCF) extension \cite{HASELI2026112722}, which parametrizes the Koopman Operator by the control input. Specifically, given a controlled system \eqref{eqn_sys}, for a fixed $u \in \mathcal{U} $, we associate a corresponding Koopman Operator 
\begin{equation}\label{eq: KCF}
\Kop_u \psi := \psi \circ f(x,u)
\end{equation}
The family of operators $\{\Kop_u\}_{u \in \mathcal{U}}$ is termed the KCF in the literature, and it describes how the observables evolve under different inputs. In conventional Koopman-MPC, first proposed by \cite{korda2018linear}, one typically works with a higher-dimensional finite set of observables $\Psi(x):=\left[\psi_1(x), \cdots{}, \psi_{N}(x)\right]^\top$, where $N \gg n_x$, and fits an affine-in-control predictor of the form
\begin{equation}\label{eq: affine in input}
\Kop_u \Psi(x) \approx A\Psi(x) + Bu
\end{equation}
where $\Kop_u$ acts on $\Psi$ component-wise. This finite-dimensional lifted predictor is commonly used in Koopman-MPC because finite-horizon propagation and elimination of the lifted states yield a QP in the planned input sequence. 

\begin{remark}\label{remark: structural bias}
The affine-in-control-input predictor \eqref{eq: affine in input} is a modeling approximation that contains structural bias and is not a universally exact controlled Koopman representation, which also requires observables involving input and mixed state-input terms. Increasing the number of observables \(N\) does not, in general, remove this structural restriction. This restriction is used because it ensures that the future state prediction is linear in $u$ and that the MPC objective is quadratic in $u$, thereby enabling QP-based MPC. See \cite{iacob2024koopman,heeg2026limitations,shang2026existence} for further details.
\end{remark}
In conventional single-step EDMD, we assume access to sampled data 
\begin{equation}
\mathcal{D}_1 = \{(x_j, u_j, x_{j,+})\}_{j=1}^{M_1}, \quad  x_{j,+} = f(x_j, u_j)
\end{equation} is 
where the subscript $+$ denotes the value at the next time step. Single-step EDMD then obtains an approximation for the state-space matrices $\{A,B\}$ by solving the following minimization problem:
\begin{equation}\label{problem_EDMD_reduced}
    J(A,B) = \min_{A,B} \sum_{j=1}^{M_1}\|\Psi(x_{j,+})-A\Psi(x_{j})-Bu_j\|_2^2.
\end{equation}
To recover the states $x$ from the lifted coordinates $\psi(x)$, one introduces an output map $C \in \mathbb{R}^{n_x \times N}$ satisfying the approximation $x \approx C \psi(x)$. The output matrix $C$ can be obtained by a similar least-squares minimization problem:
\begin{equation}\label{problem_EDMD_reduced}
    J(C) = \min_{C} \sum_{j=1}^{M_1}\|x_j - C\Psi(x_{j})\|_2^2.
\end{equation}
If the dictionary $\Psi(x)$ contains the state $x$ after re-ordering, then the output matrix is $C=[I,0]$.
The learned linear Koopman predictor is then given by
\begin{equation}\label{eqn_Koopman_linear}
    \widehat z_{k+1}
    =
     A\widehat z_k
    +
     B u_k,
    \qquad
    \widehat x_{k+1}
    =
    C\widehat z_{k+1},
\end{equation}
where $\widehat z_k\in\mathbb R^N$ denotes the predicted lifted state, with
$\widehat z_0=\Psi(x(t))$.

\subsection{Condensed Koopman-MPC formulation}
\revadd{Notably, }although the Koopman operator \revadd{learned by EDMD }may require a large lifted dimension \revadd{$N \gg n_x$} to obtain an accurate linear approximation, this does not increase the dimension of the resulting QP problem. Specifically, by recursively applying \eqref{eqn_Koopman_linear} to eliminate the lifted states $\psi_k$, we obtain 
{\small
\begin{equation}\label{eqn_X_U_E_F}
\left[\begin{array}{c}
             x_{t+1}  \\
             x_{t+2} \\
             \vdots \\
             x_{t+H}
\end{array}\right] \!= \mathbf{E}^{ss} \Psi(x_t) + \mathbf{F}^{ss} \!\left[\begin{array}{c}
     u_{t}  \\
     u_{t+1} \\
     \vdots \\
     u_{t+H-1}
\end{array}\right]
\end{equation}
where
\begin{equation}\label{eqn_E_F_def}
 \mathbf{E}^{ss} :=\! \left[\begin{array}{c}
    CA  \\
    CA^2 \\
    \vdots \\
    CA^H
\end{array}\right]\!,\ \ \mathbf{F}^{ss} := \!\left[\begin{array}{cccc}
     CB & 0 & \cdots & 0  \\
     CAB & CB & \cdots & 0 \\
     \vdots & \vdots & \ddots & \vdots \\
     CA^{H-1}B &  CA^{H-2}B & \cdots & CB
\end{array}\right]       
\end{equation}
}are the condensed matrices obtained from recursively applying single-step EDMD, and $H$ is the prediction horizon. By embedding \eqref{eqn_X_U_E_F} into the quadratic objective \eqref{subeq:nmpc-cost} and the state constraint \eqref{subeq:mpc-state-constraints}, the nonlinear MPC \eqref{eqn_nmpc} can be reduced to the condensed QP problem in the decision vector $U$ shown in \eqref{eq: condensed MPC} \cite{rawlings2020model}. 

\begin{samepage}
\noindent\rule{\columnwidth}{0.4pt}
\textbf{Condensed Koopman-MPC $\rightarrow$ General QP}
\begin{subequations}\label{eq: condensed MPC}
\begin{align}
\min_{U} \ & U^\top \bar{R}U + (\mathbf{E}\psi(x_t)+\mathbf{F}U)^\top\bar{Q} (\mathbf{E}\psi(x_t)+\mathbf{F}U)  \\
\text{s.t.} \ & \mathbf1^H \otimes x_{\min} \leq \mathbf{E}\Psi(x_t)+\mathbf{F}U \leq \mathbf1^H \otimes x_{\max}  \\
& \mathbf1^H \otimes u_{\min}\leq U \leq \mathbf1^H \otimes u_{\max}
\end{align}
\end{subequations}
\noindent\rule{\columnwidth}{0.4pt}
where $\bar{R}=\mathrm{blkdiag}(R,\cdots{},R)$
and $\bar{Q}=\mathrm{blkdiag}(Q,\cdots{},Q,P)$.
\par
\end{samepage}

The higher-dimensional observable $\Psi(x_t)$ is computed only once per QP problem \revadd{(i.e., per sampling time)}.
By performing the condensing procedure described above, the computational burden of having a high-dimensional Koopman lifting is substantially minimized when solving the MPC problem.

\section{Multi-step Extended Dynamic Mode Decomposition}\label{sec: multistep}
\revdel{Data-driven Koopman operators inevitably introduce approximation errors, and the key to successfully integrating Koopman operators with MPC is ensuring that these approximation errors do not accumulate along the prediction horizon. However, the EDMD approach in Subsection \ref{subsec:EDMD}, particularly \eqref{problem_EDMD_reduced}, is based on minimizing single-step-ahead prediction error, and could potentially suffer from accumulated approximation errors, especially in scenarios with a long prediction horizon.} 
\revadd{The preceding subsection shows that Koopman-MPC ultimately uses the finite-horizon prediction matrices \(\mathbf E\) and \(\mathbf F\). Conventional EDMD obtains these matrices indirectly: it first identifies the single-step lifted predictor matrices $\{A,B,C\}$, and then recursively propagates this model over the MPC horizon. This aligns the identification objective with single-step prediction, even though MPC uses finite-horizon predictions.} 
Consequently, the single-step-ahead EDMD approach \eqref{problem_EDMD_reduced} can be modified to incorporate multi-step prediction error:
\begin{equation}\label{eqn_EDMD_ms}
\min_{A,B}\sum_{j=1}^{\revdel{M_H}\revadd{M_H}}\sum_{k=1}^{H}\left\|\psi(x_{j,k})-A^k\Psi(x_{j,0})-\sum_{m=0}^{k-1} A^{k-1-m}Bu_{j,m} \right\|_2^2    ,
\end{equation}
where $H$ is the prediction horizon and $M_\revadd{H}$ is the number of sampled trajectories:
\[
\mathcal{D}_H= \left\{(x_{j,k},x_{j,0},u_{j,k-1}) \right\}\revadd{_{j=1}^{M_H}},~\forall j\in[1, M_\revadd{H}],~\forall k \in [1, H]
\]
($x_{j,k}$ and $u_{j,k}$ denote the state and control input at the $k$-th sampling time of the $j$-th trajectory). However, \eqref{eqn_EDMD_ms} is neither a convex nor a least-squares optimization problem, as compared to \eqref{problem_EDMD_reduced}. 

This article proposes a least-squares-based multi-step EDMD approach for real-time MPC applications. Inspired by the observation that the Koopman-MPC framework ultimately uses the representation \eqref{eqn_X_U_E_F} (the matrices $\mathbf{E}$ and $\mathbf{F}$), rather than the intermediate matrices $A$ and $B$, our proposed multi-step EDMD approach directly learns the matrices $\mathbf{E}$ and $\mathbf{F}$ in \eqref{eqn_X_U_E_F} from the data $\mathcal{D}_H$. Specifically, we first define the causal finite-horizon matrix class
\[
\mathcal{F}_H^{\mathrm{causal}}
:=
\left\{
\bF\in\mathbb{R}^{Hn_x\times Hn_u}
:
F_{k,\ell}=0 \text{ for all } \ell\geq k
\right\},
\]
where \(F_{k,\ell}\in\mathbb{R}^{n_x\times n_u}\) is the matrix block that maps \(u_\ell\) to the prediction of \(x_k\). We then obtain $\mathbf{E}$ and $\mathbf{F}$ from the multi-step prediction error minimization problem:
\begin{equation}\label{eqn:multisteperror}
\min_{\substack{
\bE\in\mathbb{R}^{Hn_x\times N}\\
\bF\in\mathcal{F}_H^{\mathrm{causal}}
}} \sum_{j=1}^{\revdel{M_H}\revadd{M_H}} \bigg\lVert X_j -  \bE \Psi(x_{j,0}) - \bF U_{j}\bigg\rVert_2^2  
\end{equation}
where $X_j = \mathrm{vec}[x_{j,1},\cdots{}, x_{j,H}]$ and $U_j = \mathrm{vec}[u_{j,0},\cdots{}, u_{j,H-1}]$ are trajectory samples from the data set $\mathcal{D}_H$. This reformulated problem is now a convex least-squares optimization that is further parallelizable at the prediction-horizon and state-coordinate levels. The parallel decomposition of this problem is introduced primarily to enforce the causal lower-triangular structure of $\bF$ and to enable row-wise $\ell_1$ regularization for observable dictionary pruning.

\begin{remark}
By directly learning the finite-horizon prediction matrices \(\{\mathbf E,\mathbf F\}\), our method relaxes the requirement that these matrices be consistent with a lifted linear time-invariant model. This intentional relaxation allows the multi-step EDMD predictor to act as a linear horizon-varying model, thereby increasing its flexibility for finite-horizon approximation while preserving the condensed linear structure required for QP-based MPC. However, independently parameterizing the horizon-wise predictors may introduce structural inconsistency across prediction horizons.
This structural inconsistency is distinct from statistical estimation error, that is, the discrepancy between the finite-sample empirical estimator and the population error minimizer. 
This cross-horizon inconsistency need not vanish with increased number of sampled trajectories, or even in the infinite-data limit.
Thus, our proposed formulation reveals a trade-off between dynamical consistency and finite-horizon prediction accuracy while preserving the condensed linear structure required for QP-based MPC.
\end{remark}

\subsection{Decomposition at the prediction-horizon and state-coordinate levels}

To enforce the structural constraint that $\bF$ is lower triangular, thereby ensuring causality of the predictor, this article adopts a simple trick: learning $E_k$ and $F_k$, namely the $k$-th row-blocks of $\bE$ and $\bF$ defined in \eqref{eqn:multisteperror}, independently and in parallel. These blocks approximate $x_k$ using $\Psi(x_0)$ and the admissible input sequence $u_0,\cdots,u_{k-1}$. To do so, we divide the full dataset $\mathcal{D}_H$ into $H$ pieces, each corresponding to one horizon step $k\in [1,H]$. Thus, a primary role of the horizon-wise decomposition is structure preservation: it explicitly ensures that the prediction at time $k$ depends only on the admissible past inputs $u_0,\ldots,u_{k-1}$, while also producing subproblems that can be solved independently. For $k=1,\ldots,H$, denote the $k$-th horizon-step data as follows:

\[
\mathcal{D}\revdel{_k}\revadd{_{k,H}}=\{(x_{j,k},x_{j,0},u_{j,0},\cdots{},u_{j,k-1})\}_{j=1}^{M_{\revdel{m}\revadd{H}}}.
\]
We can then solve the least-squares optimization problems in parallel:
\begin{equation}\label{eqn_min_E_k_F_k}
\begin{aligned}
    &\textbf{for}~k=1:H \textbf{ (in parallel)}\\
    &\quad \min_{E_k, F_k'}\sum_{\mathcal{D}\revdel{_k}\revadd{_{k,H}}} \left\|x_{j,k}-E_k \Psi(x_{j,0}) -F_k' \!\left[\begin{array}{c}
         u_{j,0}  \\
         \vdots\\
         u_{j,k-1}
    \end{array}\right]  \right\|_2^2\\
    &\textbf{end}
\end{aligned}
\end{equation}
where $F_k'\in \mathbb R^{n_x \times n_u  k}$ ensures that $\mathbf F$ is block lower-triangular. We append $F_k'$ with columns of 0 to obtain the full row block $F_k \in \mathbb R^{n_x \times n_u  H}$. Note that this decomposition does not introduce an additional approximation error relative to \eqref{eqn:multisteperror}; it is the same least-squares objective written in a horizon-wise form with the causal structure of $\bF$ imposed explicitly. For the state coordinate-decomposition, denoting $x_{j,k}^i$ as the $i$-th element of $x_{j,k}$, 
because matrix-vector multiplication is separable across rows, \eqref{eqn_min_E_k_F_k} can be equivalently decomposed into parallel subproblems.
\begin{equation}\label{eqn_min_E_k_i_F_k_i}
\begin{aligned}
    &\textbf{for}~i=1:n_x , \, k = 1:H \textbf{ (in parallel)}\\
    &\quad \min_{E_{k,i}, F_{k,i}'} \sum_{\mathcal{D}\revdel{_k}\revadd{_{k,H}}}\left(x_{j,k}^i- E_{k,i}^\top \Psi(x_{j,0})- F_{k,i}^{'\top} \!\left[\begin{array}{c}
         u_{j,0}  \\
         \vdots\\
         u_{j,k-1}
    \end{array}\right]  \right)^2\\
    &\textbf{end}
\end{aligned}
\end{equation}
where \revdel{$x_{j,k}^i$ denotes the $i$-th element of $x_{j,k}$,}  $E_{k,i}$ and $F_{k,i}'$ denote the transposes of $i$-th row vectors of matrices $E_k$ and $F_k'$, respectively.

\begin{remark}
(Exact separability of the causal multi-step EDMD regression) The proposed horizon-wise and coordinate-wise decomposition is an exact
reformulation of the least-squares problem \eqref{eqn:multisteperror}, rather than an approximation. Specifically, under the block lower-triangular causality
structure of $~\mathbf{F}$, the global least-squares objective can be written as the sum of independent horizon-wise objectives, each of which
further decomposes into independent coordinate-wise regressions because the squared Euclidean loss is additive across state coordinates. Therefore, solving \eqref{eqn_min_E_k_i_F_k_i} independently and assembling the resulting coefficients yields exactly the same minimizer as solving the original causal least-squares problem \eqref{eqn:multisteperror}. The proposed parallel implementation thus changes only the computational strategy, not the optimization problem itself, and introduces no additional approximation error.
\end{remark}

\begin{remark}
Although the optimization variables are solved independently after
decomposition, the learned coefficients are not statistically
independent. Each regression still uses the same training trajectories
and shared observable dictionary, so the correlations encoded in the
data are preserved. The decomposition is algorithmic rather than statistical: it decomposes the optimization problem but does not alter the underlying regression model or discard correlations present in the data.
\end{remark}

Note that the coordinate-wise decomposition also enables sparsity-promoting regularization for observable dictionary pruning. One may add the penalty $\tau \|E_{k,i}\|_1$ to the objective in \eqref{eqn_min_E_k_i_F_k_i}, where $\tau \geq 0$ is a regularization parameter. This promotes sparse dependence on the observable dictionary.

\section{Theoretical Analysis}\label{sec: error}
This section provides a comparative analysis of the learning error of conventional single-step EDMD and multi-step EDMD for long-horizon prediction. In Section \ref{subsec:EDMD}, we observed that single-step EDMD learns a single-step Koopman operator and generates long-horizon predictions by composing the linear predictor model. This single-step approach may potentially lead to the accumulation of approximation error across the prediction horizon in MPC. In contrast, our multi-step EDMD method in Section \ref{sec: multistep} learns the horizon-dependent input-output matrices, which directly circumvents the issue of error accumulation. To justify our hypothesis, we provide non-asymptotic error bounds under finite data budgets to demonstrate that multi-step EDMD yields stable long-horizon predictions. Our finite-sample analysis uses the Matrix Bernstein Inequality (Theorem 5.4.1 in \cite{Vershynin_2018}), which complements existing approaches in the EDMD literature \cite{NuskePeitzPhilippSchallerWorthmann2023,zhangzuazua2022,philipp2024variancerepresentationsconvergencerates,PHILIPP2024kernelacha,hertel2025koopmanstochasticdynamicserror}. 

\subsection{Analysis framework and assumptions}
\label{subsec:analysis framework and assumptions}
To formalize the controlled finite-horizon Koopman map, we first make the following assumption. 
\begin{assumption}
The set of feasible controls $\mathcal{U} \subseteq \mathbb{R}^{n_u}$ is compact, and the state domain $\Omega \subseteq \mathbb{R}^{n_x}$ is forward invariant, i.e. $f(\Omega , U) \subseteq \Omega$.
\end{assumption}
Let $\mu_x$ be a Borel probability measure on $\Omega$, $\mu_u$ be a sampling distribution for the inputs. Let $\cH := L^2(\Omega,\mu_x)$ be the Hilbert space of square-integrable functions with respect to $\mu_x$, equipped with the inner product
$
\langle f, g\rangle_\cH := \int_\Omega fg\; d\mu_x
$.

\begin{remark}
The Hilbert space norm admits the probabilistic representation:
\begin{equation}
\norm{f}_\cH^2 = \int_\Omega f^2\; d\mu_x = \underset{x\sim\mu_x}{\mathbb{E}}[f(x)^2].
\end{equation}
\end{remark}
For each fixed input \(u\in\mathcal{U}\), define
$f_u(x):=f(x,u)$ and the corresponding fixed-input Koopman operator
\begin{equation}\label{eqn_kcf_operator}
    \Kop_u g := g\circ f_u,\qquad g\in\cH .
\end{equation}
To ensure the KCF \(\{\Kop_u\}_{u\in\mathcal{U}}\) is well-defined, we impose the following assumption: 
\begin{assumption}\label{assum:bounded-kcf}
For any $u \in \mathcal{U}$, the operator $\Kop_u \in \mathcal{B}(\cH)$ is bounded.
\end{assumption}
Assuming the pushforward measure $(f_u)_\#\mu_x \ll \mu_x$ with uniformly bounded Radon--Nikodym derivatives $C_f := \sup_{u \in \mathcal{U}} \| \frac{d(f_u)_\#\mu_x}{d\mu_x} \|_{L^\infty} < \infty$, the Koopman operator $\mathcal{K}_u$ is bounded. Indeed, by the change of variables formula:
\begin{equation}
\begin{aligned}
\|\mathcal{K}_u g\|_{\mathcal{H}}^2 
&= \int_{\Omega} |g(f_u(x))|^2 \, d\mu_x(x) \\
&= \int_{\Omega} |g(y)|^2 \, d((f_u)_\#\mu_x)(y) \le C_f \|g\|_{\mathcal{H}}^2.
\end{aligned}
\end{equation}
Thus, $\mathcal{K}_u \in \mathcal{B}(\mathcal{H})$ for all $u$, satisfying Assumption \ref{assum:bounded-kcf}.

\subsection{Koopman approximations for EDMD}
A numerical approximation of the KCF requires restricting the space of observables to a finite-dimensional subspace. To this end, let
$
\Psi(x) := \left[\psi_1(x), \ldots, \psi_N(x)\right]^\top
 $
denote a finite dictionary of observables in $\cH^{N}$.

\subsubsection{Single-step EDMD}
We first analyze the horizon-prediction error obtained by recursively applying a single-step EDMD model. For matrices
$A\in \mathbb R^{N\times N}$ and $B\in \mathbb R^{N\times n_u}$, define the population single-step prediction error
\begin{equation}\label{eq:population-risk-ss}
    \mathbf{Err}^{ss}(A,B)
    :=
    \underset{(x,u)\sim \mu_x\otimes \mu_u}{\mathbb{E}}
    \left[
        \left\|
        \Kop_u \Psi(x)
        -
        A\Psi(x)
        -
        Bu
        \right\|_2^2
    \right].
\end{equation}
Importantly, note that the expectation in \eqref{eq:population-risk-ss} is finite since
$\Kop_u\Psi\in \mathcal H^N$ by Assumption~\ref{assum:bounded-kcf}. Let
\begin{equation}\label{eq:ss-pop-minimizer}
    (A_*,B_*)\in 
    \operatorname*{arg\,min}_{A,B}
    \mathbf{Err}^{ss}(A,B)
\end{equation}
denote a minimizer of the population single-step prediction error. In single-step EDMD, $(A_*,B_*)$ is estimated from sampled transition data. Specifically, let
$\{(x_j,u_j)\}_{j=1}^M$ be i.i.d. samples, where
$(x_j,u_j)\sim \mu_x\otimes \mu_u$. The empirical single-step EDMD loss is
\begin{equation}\label{eq:empirical-risk-ss}
    \widehat{\mathbf{Err}}^{ss}_M(A,B)
    :=
    \frac{1}{M}
    \sum_{j=1}^M
    \left\|
        \Kop_u \Psi(x_j)
        -
        A\Psi(x_j)
        -
        Bu_j
    \right\|_2^2 .
\end{equation}
The single-step EDMD estimator is then defined by
\begin{equation}\label{eq:ss-edmd-estimator}
    (\widehat A_M,\widehat B_M)
    \in
    \operatorname*{arg\,min}_{A,B}
    \widehat{\mathbf{Err}}^{ss}_M(A,B).
\end{equation}
Once the EDMD matrices $(\widehat A_M,\widehat B_M)$ are computed, it can be used to generate predictions of the dynamics.

\textbf{Regression Lemma. } The estimation error in least-squares EDMD algorithms reduces to standard linear regression problems, as seen in \eqref{eqn_min_E_k_F_k} and \eqref{eq:empirical-risk-ss}. We state a high-probability regression bound below that will be used in the error analysis of both single-step and multi-step EDMD. 

\begin{lemma}[Regression Error]\label{lem:regression}
Let $\{(\chi_j,Y_j)\}_{j=1}^M$ be i.i.d.\ samples of a random pair
$(\chi,Y)$, where $\chi\in\mathbb R^{N_\chi}$ and $Y\in\mathbb R$. Assume that
$\mathbb E[\chi\chi^\top]=I_{N_\chi}$, and
    $    \|\chi\|_2\le B_\chi$,
    $|Y|\le M_y$ almost surely. Define the population minimizer
\[
    \beta_\star
    \in
    \operatorname*{arg\,min}_{\beta\in\mathbb R^{N_\chi}}
    \mathbb E\left[
        \left|Y-\beta^\top\chi\right|^2
    \right].
\]
Define the empirical quantities
\[
    \widehat\Sigma
    :=
    \frac{1}{M}\sum_{j=1}^M \chi_j\chi_j^\top,
    \qquad
    \widehat g
    :=
    \frac{1}{M}\sum_{j=1}^M \chi_jY_j.
\]
Let
$
    \widehat\beta := \widehat\Sigma^{-1}\widehat g.
$
Then, if $M \gg B_\chi^2 \log M$, with probability at least $1-2M^{-2}$,
\begin{equation}\label{eq:regression-error}
    \mathbb E\left[
        \left|Y-\widehat\beta^\top\chi\right|^2
    \right]
    -
    \mathbb E\left[
        \left|Y-\beta_\star^\top\chi\right|^2
    \right]
    \le
    C_l\,\frac{M_y^2 B_\chi^2\log M}{M},
\end{equation}
for a universal constant $C_l>0$.
\end{lemma}

\begin{proof}
The proof is deferred to Appendix~\ref{app:regression}.
\end{proof}

We now apply Lemma~\ref{lem:regression} to analyze the approximation error of conventional single-step EDMD. To apply Lemma~\ref{lem:regression}, we impose the following mild assumptions on the observables and inputs.

\begin{assumption}\label{assum:ss-features}
The following conditions hold:
\begin{itemize}
    \item Each observable satisfies $\mathbb{E}_{\mu_x}[\psi_i(x)] = 0$ and is uniformly bounded, i.e.,
    $|\psi_i(x)|\le B_\psi$ almost surely for all $i=1,\ldots,N$ and for some constant $B_\psi>0$.
            \item The inputs satisfy $\mathbb{E}_{\mu_{u}}[ u] = 0$ and 
    $\|u\|_2\le R_u$ almost surely for some constant $R_u>0$.
    \item The joint observable-input feature vector
    $
    \chi(x,u):=
    [\Psi(x)^\top, \; u^\top]^\top
    \in \mathbb R^{N+n_u}
    $
    satisfies
    \[
        \Sigma_\chi
        :=
        \mathbb E[\chi(x,u)\chi(x,u)^\top]
        =
        I_{N+n_u}.
    \]
\end{itemize}
\end{assumption} The final condition can be satisfied by centering and separately orthonormalizing the observable and input features, since $\mu_x$ and $\mu_u$ are independent.

Applying Lemma \ref{lem:regression} to the single-step EDMD algorithm then yields the following high-probability bound:
\begin{prop}[Single-step EDMD estimation error]\label{prop:ss-edmd-error}
Under Assumption~\ref{assum:ss-features}, and for
$M\gg (NB_\psi^2+R_u^2)\log M$, with probability at least $1-2NM^{-2}$, the single-step EDMD operator matrices $(\widehat A_M,\widehat B_M)$ satisfy 
\begin{align}
    \mathbf{Err}^{ss}(\widehat A_M,\widehat B_M)
    &\le
    \mathbf{Err}^{ss}(A_*,B_*) \nonumber\\
    &+ 
    C_l\,
    \frac{
        N B_\psi^2\left(NB_\psi^2+R_u^2\right)\log M
    }{M}. \label{eqn:ssEDMDerror}
\end{align}
where $~C_l>0$ is a universal constant and $M$ is the number of single-step samples.
\end{prop}
\begin{proof}
The proof is similar to the multi-step case and deferred to Appendix \ref{app:singlestep}.
\end{proof}
Equipped with the error bound \eqref{eqn:ssEDMDerror} for single-step EDMD, we can then analyze the $L^2$ error when the single-step EDMD operator is applied recursively to predict future states. For the single-step EDMD propagation analysis only, we impose the following simplifying assumptions.

\begin{assumption}\label{assum:ss-rollout}
The following conditions hold:
\begin{itemize}
    \item The state coordinates are included in the dictionary, so
    $
        x=C\Psi(x)
    $
    with
    $
        C=\begin{bmatrix}I_{n_x} & 0\end{bmatrix}.
    $
    \item The sampling measure $\mu_x$ is invariant under the randomized controlled dynamics induced by $u\sim\mu_u$; equivalently, if $x\sim\mu_x$ and $u\sim\mu_u$, then $f(x,u) \sim\mu_x$.
\end{itemize}
\end{assumption}

Assumption~\ref{assum:ss-rollout} requires that the distribution $\mu_x$ be invariant under the randomized controlled dynamics induced by $u\sim\mu_u$.
For example, systems driven by i.i.d inputs on a forward-invariant domain admit an invariant probability distribution if the Lipschitz constant $L(u)$ of
$x\mapsto f(x,u)$ satisfies
$\int\log L(u)\,\mu_u(du)<0$ and the algebraic-tail conditions
in Theorem 5.1 of~\cite{diaconis1999iterated} hold.
We can sample from this stationary distribution by initializing independent experiments with $x_0\sim\mu_x$ and applying input sequences whose entries are drawn independently from $\mu_u$.

Theorem~\ref{thm:ss-error-propagation} for single-step EDMD uses this stationarity condition to ensure that the expected squared single-step residual error is evaluated under the same distribution at every prediction horizon.
This condition need not hold for arbitrary transient data; accordingly, Theorem~\ref{thm:ss-error-propagation} applies specifically to this stationary setting. 
The multi-step EDMD error analysis in the next subsection, including the proposed multi-step identification procedure and the multi-step error bound in Theorem~\ref{thm:ms-edmd-error}, does not require this invariance condition.

To compare single-step and multi-step EDMD using the same metric, we first define the finite-horizon prediction error. Let
$
U:=(u_0^\top,\ldots,u_{H-1}^\top)^\top \sim \mu_u^{\otimes H}
$
be the control sequence with entries drawn independently from $\mu_u$, and let
$
X_H(x_0,U):=(x_1^\top,\ldots,x_H^\top)^\top
$
denote the stacked state trajectory generated from the initial condition $x_0$ and input sequence $U$. For matrices
$\mathbf E\in\mathbb R^{Hn_x\times N}$ and
$\mathbf F\in\mathbb R^{Hn_x\times Hn_u}$, define the finite-horizon multi-step prediction error:
\begin{equation}\label{eq:horizon-risk}
    \mathbf{Err}^{ms}(\mathbf E,\mathbf F)
    :=
    \mathbb E_{\mu_H}
    \left[
        \left\|
        X_H(x_0,U)
        -
        \mathbf E\Psi(x_0)
        -
        \mathbf F U
        \right\|_2^2
    \right],
\end{equation}
where $\mu_H$ denotes the joint sampling distribution of $(x_0,U)$.
As shown in Section \ref{sec: multistep}, a recursively applied single-step EDMD model also recursively defines a condensed multi-step predictor of the form
$
\mathbf E^{ss}(\widehat A_{M},C)\Psi(x_0)+\mathbf F^{ss}(\widehat A_{M},\widehat B_{M},C) U
$, where $\mathbf E^{ss}$ and $\mathbf F^{ss}$ are defined in \eqref{eqn_E_F_def}.

To quantify the multi-step prediction error of recursively applied single-step EDMD, define the learned single-step error
\[
    \widehat r^{ss}(x,u)
    :=
    \Kop_u\Psi(x)
    -
    \widehat A_M\Psi(x)
    -
    \widehat B_Mu .
\]

\begin{theorem}\label{thm:ss-error-propagation}
Suppose Assumptions~\ref{assum:ss-features} and~\ref{assum:ss-rollout} hold. Then, for
$M\gg (NB_\psi^2+R_u^2)\log M$, with probability at least $1-2NM^{-2}$, the finite-horizon prediction error of recursively applied single-step EDMD satisfies
\begin{align}\label{eq:ss-ms-propagation-bound}
    \mathbf{Err}^{ms}(\mathbf E^{ss},\mathbf F^{ss})
    &\le
    \Gamma_H(\widehat A_M,C)
    \Bigg[
        \mathbf{Err}^{ss}(A_*,B_*) \nonumber \\
        &
        +
        C_l
        \frac{
            N B_\psi^2\left(NB_\psi^2+R_u^2\right)\log M
        }{M}
    \Bigg],
\end{align}
where
\begin{equation}\label{eq:gamma-ss}
    \Gamma_H(\widehat A_M,C)
    :=
    \sum_{t=1}^H
    \left(
        \sum_{\ell=0}^{t-1}
        \left\|
            C\widehat A_M^{\,t-1-\ell}
        \right\|_2
    \right)^2,
\end{equation}
$C_l>0$ is a universal constant, and $M$ is the number of single-step samples. 
\end{theorem}

\begin{proof}
Using the notation in \eqref{eqn_Koopman_linear}, let $\widehat z_t = \widehat A_M^t \Psi(x_0) + \sum_{l=0}^{t-1} \widehat A_M^{t-1-l} Bu_{l}  $ be the single-step prediction in observable space, and let the observable error be $e_t:=\Psi(x_t)-\widehat z_t$. Since $x=C\Psi(x)$ and
$\widehat x_t=C\widehat z_t$, the state prediction error is
$x_t-\widehat x_t=Ce_t$. Moreover,
\[
    e_{t+1}
    =
    \widehat A_Me_t+\widehat r^{ss}(x_t,u_t),
    \qquad e_0=0.
\]
Iterating this recursion gives
\[
    e_t
    =
    \sum_{\ell=0}^{t-1}
    \widehat A_M^{t-1-\ell}
    \widehat r^{ss}(x_\ell,u_\ell).
\]
Therefore,
\[
    \|x_t-\widehat x_t\|_2
    \le
    \sum_{\ell=0}^{t-1}
    \left\|
        C\widehat A_M^{t-1-\ell}
    \right\|_2
    \left\|
        \widehat r^{ss}(x_\ell,u_\ell)
    \right\|_2 .
\]
Under Assumption~\ref{assum:ss-rollout}, if $x_0\sim\mu_x$ and
$u_t\sim\mu_u$ independently, then $(x_t,u_t)\sim\mu_x\otimes\mu_u$ for all
$t$. Hence
\[
    \mathbb E
    \left[
        \left\|
            \widehat r^{ss}(x_t,u_t)
        \right\|_2^2
    \right]
    =
    \mathbf{Err}^{ss}(\widehat A_M,\widehat B_M).
\]
Applying Cauchy's inequality and summing over $t=1,\ldots,H$ yields
\[
    \mathbf{Err}^{ms}(\mathbf E^{ss},\mathbf F^{ss})
    \le
    \Gamma_H(\widehat A_M,C)
    \mathbf{Err}^{ss}(\widehat A_M,\widehat B_M).
\]
The result follows by applying Proposition~\ref{prop:ss-edmd-error}.
\end{proof}
Theorem~\ref{thm:ss-error-propagation} explicitly highlights the accumulation of single-step residuals across the prediction horizon. Specifically, if
$\|\widehat A_M\|_2>1$, the residuals at earlier time steps
$\widehat r^{ss}(x_\ell,u_\ell)$ are amplified by the factor
$\|C\widehat A_M^{\,t-1-\ell}\|_2$ when propagated to the prediction error at time $t$. This motivates the formulation of multi-step EDMD over a finite prediction horizon.

\subsection{Multi-step EDMD}
We now analyze the multi-step EDMD algorithm proposed in Section \ref{sec: multistep}. For a fixed dictionary of observables $\Psi(x):=\left[\psi_1(x), \cdots{}, \psi_{N}(x)\right]^\top$, we define the population multi-step prediction error for horizon matrices $\bE$ and $\bF$ in \eqref{eq:horizon-risk}.

Let 
\begin{equation}
(\bE_\star,\bF_\star)
\in
\operatorname*{arg\,min}_{\substack{
\bE\in\mathbb{R}^{Hn_x\times N}\\
\bF\in\mathcal{F}_H^{\mathrm{causal}}
}}
\mathbf{Err}^{ms}(\bE,\bF)
\end{equation}
denote a minimizer of the population multi-step prediction error over the causal hypothesis class.

The residual error $\textbf{Err}^{ms}(\bE_*,\bF_*)$ contains both the dictionary approximation error and structural bias induced by the affine-in-control-input restriction (see Remark \ref{remark: structural bias}).

Formally, the multi-step EDMD algorithm estimates the population minimizers 
\((\mathbf E_\star,\mathbf F_\star)\) from sampled trajectory data. Specifically, 
we assume access to \(M_H\) i.i.d. length-\(H\) trajectory samples
\[
    \mathcal D_H
    :=
    \{(x_{j,0},U_j,X_j)\}_{j=1}^{M_H},
    \qquad
    (x_{j,0},U_j)\sim \mu_H,
\]
where $\mu_H = \mu_x \otimes \mu_u^{\otimes H}$. The multi-step EDMD estimator then constructs the approximation 
\((\widehat{\mathbf E},\widehat{\mathbf F})\) by solving the empirical finite-horizon loss

\begin{equation}\label{eqn:emplosstraj}
    (\widehat{\mathbf E},\widehat{\mathbf F})
    \in
    \operatorname*{arg\,min}_{\substack{
\bE\in\mathbb{R}^{Hn_x\times N}\\
\bF\in\mathcal{F}_H^{\mathrm{causal}}
}}
    \frac{1}{M_H}
    \sum_{j=1}^{M_H}
    \left\|
        X_j
        -
        \mathbf E\Psi(x_{j,0})
        -
        \mathbf F U_j
    \right\|_2^2,
\end{equation}
The population benchmark in Equation (31) and the empirical estimator in Equation (32) are therefore defined over the same causal hypothesis class, with \(\mathbf F\) restricted to be block lower triangular in both problems.

Following the same reasoning as the single-step case, Lemma \ref{lem:regression} can be similarly applied to analyze the approximation error of multi-step EDMD and compare it with \eqref{eqn:ssEDMDerror}. In particular, we have the following high-probability bound:

\begin{theorem}[Multi-step EDMD estimation error]\label{thm:ms-edmd-error}
Suppose Assumption~\ref{assum:ss-features} holds and the state space $\Omega$ satisfies $\sup_{x \in \Omega} \norm{x} \leq B_x$.
Then, for $M_H \gg  \left(NB_\psi^2+HR_u^2\right) \log M_H$, with probability at least $1-2Hn_xM_H^{-2}$, the multi-step EDMD matrices $(\widehat{\mathbf E},\widehat{\mathbf F})$ satisfy 
\begin{align}
    \mathbf{Err}^{ms}(\widehat{\mathbf E},\widehat{\mathbf F})
    &\le
    \mathbf{Err}^{ms}(\mathbf E_\star,\mathbf F_\star) \nonumber \\
    &+
    C_l\,
    \frac{
        Hn_x\,B_x^2
        \left(NB_\psi^2+HR_u^2\right)
        \log M_H
    }{M_H}. \label{eq:ms-edmd-error}
\end{align}
where $C_l>0$ is a universal constant and $M_H$ is the number of sampled trajectories in $\mathcal{D}_{H}$.
\end{theorem}

\begin{proof}
For each $t=1,\ldots,H$, let $\bE_t$ denote the $t$-th block row of
$\mathbf E$, and let $\mathbf F_{t,\ell}$ denote the block multiplying $u_\ell$
in the prediction of $x_t$. Since $\mathbf F$ is block lower triangular, the
$t$-step prediction has the form
\[
    \widehat x_t
    =
    \mathbf E_t\Psi(x_0)
    +
    \sum_{\ell=0}^{t-1}
    \mathbf F_{t,\ell}u_\ell .
\]
Define the corresponding feature vector
\[
    \chi_t(x_0,U)
    :=
    [\Psi(x_0)^\top\; u_0^\top\;\cdots\;u_{t-1}^\top]^\top .
\]
For each coordinate $i=1,\ldots,n_x$, let $[\mathbf E_t]_i$ and
$[\mathbf F_{t,\ell}]_i$ denote the $i$-th rows of $\mathbf E_t$ and
$\mathbf F_{t,\ell}$, respectively, and define
\[
    \theta_{t,i}
    :=
    \left[
        [\mathbf E_t]_i \;\;
        [\mathbf F_{t,0}]_i \;\;
        \cdots \;\;
        [\mathbf F_{t,t-1}]_i
    \right]^\top .
\]
Then the $i$-th coordinate of the $t$-step predictor is
$
    [\widehat x_t]_i
    =
    \theta_{t,i}^\top \chi_t(x_0,U).
$
Let
$    Y_{t,i}(x_0,U):=[x_t]_i .
$
The scalar population error for the $i$-th state coordinate in the $t$-th block row is
\begin{equation}\label{eq:ms-scalar-risk}
    \mathcal R_{t,i}(\theta)
    :=
    \mathbb E_{\mu_H}
    \left[
        \left|
            Y_{t,i}(x_0,U)
            -
            \theta^\top\chi_t(x_0,U)
        \right|^2
    \right].
\end{equation}
The corresponding population minimizer is
\[
    \theta_{t,i,*}
    \in
    \operatorname*{arg\,min}_{\theta}
    \mathcal R_{t,i}(\theta),
\]
and the empirical minimizer is
\[
    \widehat\theta_{t,i}
    \in
    \operatorname*{arg\,min}_{\theta}
    \frac{1}{M_H}
    \sum_{j=1}^{M_H}
    \left|
        Y_{j,t,i}
        -
        \theta^\top\chi_t(x_{j,0},U_j)
    \right|^2 .
\]
where $Y_{j,t,i} = [x_{j,t}]_i$ denotes the $i$-th coordinate of $x_{j,t}$. By the row-wise decomposition of the least-squares loss, the collection of
$\widehat\theta_{t,i}$ by concatenating each $\widehat\theta_{t,i}$ row-wise is exactly the estimator
$(\widehat{\mathbf E},\widehat{\mathbf F})$. Similarly, the collection of
$\theta_{t,i,*}$ by concatenating each $\theta_{t,i,*}$  row-wise is the population minimizer
$(\mathbf E_\star,\mathbf F_\star)$.

We now verify the assumptions of Lemma~\ref{lem:regression}. Since
$\mu_H=\mu_x\otimes\mu_u^{\otimes H}$, the controls are i.i.d. and independent
of $x_0$. By Assumption~\ref{assum:ss-features}, the observable features and
inputs are centered, and
$
    \mathbb E[\Psi(x_0)\Psi(x_0)^\top]=I_N,
    \mathbb E[u_\ell u_\ell^\top]=I_{n_u}.
$
The cross-moments are zero because $\Psi(x_0)$ and $u_l$ for $l = 0,1,\cdots ,t-1$ are jointly independent and $\mathbb{E}[\Psi(x_0)] = \mathbb{E}[u_l] = 0$. Hence
$
    \mathbb E[\chi_t\chi_t^\top]=I_{N+tn_u}.
$
Moreover,
\begin{align*}
    \|\chi_t(x_0,U)\|_2^2
    &=
    \|\Psi(x_0)\|_2^2
    +
    \sum_{\ell=0}^{t-1}\|u_\ell\|_2^2
    \\
    &\le
    NB_\psi^2+tR_u^2
    \le
    NB_\psi^2+HR_u^2 .
\end{align*}
Thus Lemma~\ref{lem:regression} applies with
$
    B_\chi^2=NB_\psi^2+HR_u^2.
$
By the assumption that the state space $\Omega$ is forward-invariant and $\sup_{x \in \Omega} \norm{x} \leq B_x$, we also have 
$
    |Y_{t,i}(x_0,U)|
    =
    |[x_t]_i|
    \le
    B_x
$. Therefore Lemma~\ref{lem:regression} applies with
$M_y=B_x$. For each fixed pair $(t,i)$, Lemma~\ref{lem:regression} gives, with probability
at least $1-2M_H^{-2}$,
\[
    \mathcal R_{t,i}(\widehat\theta_{t,i})
    \le
    \mathcal R_{t,i}(\theta_{t,i,*})
    +
    C_l
    \frac{
        B_x^2
        \left(NB_\psi^2+HR_u^2\right)
        \log M_H
    }{M_H}.
\]
Taking a union bound over all $t=1,\ldots,H$ and $i=1,\ldots,n_x$, the above
inequalities hold simultaneously with probability at least
$1-2Hn_xM_H^{-2}$. Summing over all coordinates and all prediction times gives \eqref{eq:ms-edmd-error}.
\end{proof}

\subsection{\revadd{Comparison of single-step and multi-step EDMD}}
The error bounds \eqref{eq:ss-ms-propagation-bound} and \eqref{eq:ms-edmd-error} show that learning the condensed formulation in multi-step EDMD fundamentally alters the way that the learning error propagates through the prediction horizon. Specifically, in conventional single-step EDMD, \eqref{eq:ss-ms-propagation-bound} shows that the learning error accumulates due to the repeated application of the learned single-step matrix $\widehat A_M$. This is reflected by the factor $\|C\widehat A_M^{\,t-1-\ell}\|_2$ that amplifies the single-step residual $\widehat r^{ss}(x_\ell,u_\ell)$ when propagated to the prediction error at time $t$. In the event that $\|\widehat A_M\|_2>1$ or the spectral radius $\rho(\widehat A_M)>1$, even low single-step errors can lead to significant degradation in long-horizon predictions as a result of this amplification. This phenomenon is also observed in our case studies in Section~\ref{sec: numerical experiments}.

On the other hand, the condensed multi-step EDMD formulation directly approximates the finite-horizon input-output mapping $(x_0,U)\mapsto X_H(x_0,U)$, thereby avoiding the artificial accumulation of error due to repeated composition of the learned single-step operator. As a result, the error bound in \eqref{eq:ms-edmd-error} depends on the approximation quality of the multi-step map, captured by the structural bias $\mathbf{Err}^{ms}(\mathbf E_\star,\mathbf F_\star)$, and the number of data trajectories $M_H$, rather than on powers of $\widehat A_M$. This fundamental structural difference in the learning method explains the improved performance of multi-step EDMD in practice and motivates the usage of multi-step EDMD in long-horizon prediction and MPC.

It is important to note that in the multi-step EDMD formulation, each data trajectory effectively contributes one training sample for learning the condensed matrices. This reflects a fundamental tradeoff between expressivity and variance, since multi-step EDMD reduces error accumulation across the prediction horizon at the cost of increased data requirements. In applications where a dynamic model of the nonlinear system is available, this is typically not restrictive as data can be generated at a low cost.

\section{Practical Guidelines}\label{sec: practical}
The preceding sections and theoretical analyses provide insights into practical guidelines for multi-step EDMD applications in MPC.
\subsection{Offline learning, Online MPC}
We reiterate that the multi-step EDMD operator is learned offline before integration with MPC. During MPC, the observables $z_0 = \bpsi(x_0)$ is only computed once per Condensed Koopman-MPC problem \eqref{eq: condensed MPC}. This substantially reduces the online computational burden of having a high-dimensional Koopman lifting $N$. Moreover, the online computational complexity of QP-based Koopman-MPC using the multi-step EDMD operator is identical to that of the corresponding single-step condensed Koopman-MPC problem.

\subsection{Pre-computation in Parallelized Subproblems} The multi-step EDMD learning process can be further accelerated by parallelizing the multi-step prediction error minimization problem \eqref{eqn_min_E_k_i_F_k_i}. 
Specifically, we can write the summation in \eqref{eqn_min_E_k_i_F_k_i} as a sum over vector entries: 
\begin{equation}\label{eqn_minimization}
\begin{aligned}
    \min_{E_{k,i}, F_{k,i}'} f_{k,i}(E_{k,i},F_{k,i}')=
    \left\|\mathbf{G}E_{k,i} + \mathbf{H}_kF_{k,i}'-h_k^i\right\|_2^2 
\end{aligned}
\end{equation}
where 
\begin{small}
\setlength{\eqntopsep}{4pt}
\setlength{\parskip}{0pt}
\setlength{\abovedisplayskip}{4pt}
\setlength{\abovedisplayshortskip}{0pt}
\setlength{\topsep}{0pt}
\setlength{\partopsep}{0pt}
\begin{equation}\label{eqn:GH_def}
\mathbf{G} := \!\begin{bmatrix} \Psi(x_{1,0})^\top \\ \vdots \\ \Psi(x_{M_{\revdel{_m}\revadd{_{H}}},0})^\top \end{bmatrix}\!, \quad
\mathbf{H}_k := \!\begin{bmatrix} u_{1,0}^\top & \cdots & u_{1,k-1}^\top \\ \vdots & & \vdots \\ u_{M_{\revdel{_m}\revadd{_{H}}},0}^\top & \cdots & u_{M_{\revdel{_m}\revadd{_{H}}},k-1}^\top \end{bmatrix}\!,
\end{equation}
\end{small}
and
\begin{equation}\label{eqn:h_def}
h_k^i := \begin{bmatrix} x_{1,k}^i & \cdots & x_{M_{\revdel{_m}\revadd{_{H}}},k}^i \end{bmatrix}^\top\!.
\end{equation}
Importantly, $\mathbf{G}$ can be pre-computed only once and reused for every horizon, and $\mathbf{H}_k$ can be pre-computed once per horizon for every state coordinate. This can accelerate the learning process when the number of trajectories $M_{\revdel{_m}\revadd{_{H}}}$ is very large.

\revadd{To reduce offline computation, \eqref{eqn_minimization} can be implemented using precomputed block-Gram matrices rather than repeatedly forming $[\mathbf G,\mathbf H_k]$ and its normal matrix for each horizon $k$. Specifically, $\mathbf G^\top\mathbf G$, $\mathbf G^\top\mathbf H_H$, and $\mathbf H_H^\top\mathbf H_H$ are computed once, and the horizon-dependent products are obtained by slicing the corresponding full-horizon blocks. The right-hand sides are similarly batched and sliced. This avoids duplicate matrix products before the independent row-wise regressions are parallelized across horizons and state coordinates.}

\subsection{Elastic net regularization and pruning}
Our theoretical analysis assumed that the observables $\{\psi_i\}_{i=1}^N$ are orthonormal with respect to the sampled distribution $\mu_x$. In practice this assumption could be violated in higher-dimensional settings such as discretizations of PDE systems, where spatially neighboring states are strongly correlated. As a result, the empirical $L^2$ error regression problems may be highly ill-conditioned, which may lead to numerical instability during the learning process. To mitigate this issue, we can employ elastic-net regularization when learning the multi-step EDMD operators. Formally, the regularized problem is 
\begingroup
\small
\setlength{\eqntopsep}{6pt}
\setlength{\abovedisplayskip}{6pt}
\setlength{\belowdisplayskip}{6pt}
\setlength{\topsep}{0pt}
\setlength{\partopsep}{0pt}
\begin{equation}\label{eqn_sparse_reg_minimization}
\begin{aligned}
    & \min_{E_{k,i}, F_{k,i}'} f_{k,i}(E_{k,i},F_{k,i}')=\\
    & \left\|\mathbf{G}E_{k,i} + \mathbf{H}_kF_{k,i}'-h_k^i\right\|_2^2 + \beta\left\|\left[\begin{array}{c}
         E_{k,i}\\
         F_{k,i}'  
    \end{array}\right]\right\|_2^2 + \tau \|E_{k,i}\|_1    
\end{aligned}
\end{equation}
\endgroup
where $\beta>0$ and $\tau>0$ are ridge and $\ell_1$ regularization parameters, respectively. The ridge regularization improves numerical conditioning, while the $\ell_1$ regularization promotes sparsity of the resulting multi-step EDMD model \cite{sun2020alven,zou2005regularization}. Additionally, an optional pruning step can be performed to remove observables with negligible contribution to the learned dynamics. This step improves interpretability and reduces computation time for the multi-step EDMD operator while preserving prediction accuracy. An example of a straightforward multi-step EDMD learning algorithm with elastic net regularization and pruning is provided in Algorithm \ref{algo:L1-train}.

\begin{algorithm}
    \caption{Elastic-net-regularized multi-step EDMD with pruning of nonlinear observables}\label{algo:L1-train}
    \small
    \setlist[enumerate]{nosep}
    \textbf{Input}: Dataset $\mathcal{D}= \left\{(x_{j,k},x_{j,0},u_{j,k-1}) \right\}~\forall j\in[1,M_H],~\forall k \in [1,H]$, the initial library of nonlinear observables $\psi:\rr^{n_x}\rightarrow\rr^{N}$, regularization parameters $\beta\geq0,\tau\geq0$, the pruning parameter $\epsilon_E\geq0$, the MPC prediction horizon length $H$.
    \vspace*{.1cm}\hrule\vspace*{.1cm}
    \begin{enumerate}[label*=\arabic*., ref=\theenumi{}]
        \item Divide $\mathcal{D}$ into $H$ pieces: for $k=1,\cdots{},H$ 
        \[
        \mathcal{D}_k=\{(x_{j,k},x_{j,0},u_{j,0},\cdots{},u_{j,k-1})\}_{j=1}^{M_H};
        \]
        \item Prepare the matrix $\mathbf{G}$ \eqref{eqn:GH_def} by computing $\psi(x_{j,0})$;
        \item \textbf{for} $k=1,\ldots,H$ \textbf{(in parallel)}:
        \begin{enumerate}[label=\theenumi{}.\arabic*., ref=\theenumi{}.\arabic*]
            \item Prepare the matrix $\mathbf{H}_k$ \eqref{eqn:GH_def};
            \item \textbf{for} $i=1,\ldots,n_x$ \textbf{(in parallel)}:
            \begin{enumerate}[label=\theenumii{}.\arabic*., ref=\theenumii{}.\arabic*]
                \item Prepare the vector $h_k^i$ \eqref{eqn:h_def};
                \item $E_{k,i} , F_{k,i}' \leftarrow$  the solution of  \eqref{eqn_sparse_reg_minimization};
            \end{enumerate}
            \item Get $E_k$ by combining $E_{k,i}$: $E_k\leftarrow (\cdots,E_{k,i},\cdots)^\top$;
            \item Get $F_k$ by combining $F_{k,i}$: $F_k\leftarrow (\cdots,F_{k,i},\cdots)^\top$;
        \end{enumerate}
        \item Get $E,F$ by combining $E_k,F_k$;
        \item The index set $L\leftarrow\emptyset$;
        \item \textbf{for} $l=1,\ldots,N$: \textbf{if} $E_{l,:}$ (the $l$-th column of $E$) such that $\|E_{l,:}\|_\infty\geq\epsilon_E$, \textbf{then set} $L\leftarrow L\cup\{l\}$; \textbf{end}      
        \item Pruning $E$: $E\leftarrow E_{L,:}$
        \item Pruning the nonlinear observable $\psi$: $\psi\leftarrow\psi_{L}$;
        \item $N\leftarrow\card(L)$;
        \item \textbf{end}.
    \end{enumerate}
    \vspace*{.1cm}\hrule\vspace*{.1cm}
    \textbf{Output}: Multi-step Koopman predictor with $E,F$ matrices, and nonlinear observables $\psi$.
\end{algorithm}
Importantly, these regularization and pruning steps are optional and do not structurally change the multi-step EDMD learning method for MPC. They serve as add-ons to improve the numerical stability and scaling performance of the multi-step EDMD framework.

\subsection{Model complexity and data availability}
The error bound in \eqref{eq:ms-edmd-error} explicitly depends on the observable dimensionality $N$ and the number of data trajectories $M_H$. While increasing $N$ decreases the projection error (bias), it increases the regression error (variance), which is a concrete example of a bias-variance tradeoff. Consequently, the multi-step EDMD model performance requires balancing model expressivity with data availability, ideally with $M_H \gg N^2$. Importantly, when a dynamic model of the nonlinear system is available, multi-step EDMD is particularly effective since data trajectories can be generated readily and data scarcity is not a limiting factor.

\section{Numerical Experiments}\label{sec: numerical experiments}
This section evaluates our proposed multi-step Koopman MPC framework in two benchmark numerical case studies: the Duffing oscillator and a time-varying process model. These two illustrative examples provide insight into the performance of our framework in two distinct dynamical systems. For each numerical case study, we compare the performance of four different predictors: 
\begin{enumerate}
\item Single-step EDMD
\item Recursive Least-Squares (RLS) Multi-step Koopman \cite{sayed2024recursive}
\item Multi-step Koopman Predictor 
\item Multi-step Koopman Predictor with pruning
\end{enumerate}
These two case studies demonstrate the efficacy of our multi-step Koopman framework in two distinct dynamical regimes.
All numerical results can be reproduced with the code found in \url{https://github.com/SOLARIS-JHU/Multi-step-Koopman}.

\subsection{Duffing Oscillator}
We consider a canonical case study involving the Duffing oscillator. The dynamics of the Duffing oscillator are given by 
\begin{equation}
\begin{aligned}
\dot{x}_1 &= x_2,\qquad \dot{x}_2 = -\delta x_2 - \alpha x_1 - \beta x_1^3+ u,
\end{aligned}
\end{equation}
where $x:=[x_1,x_2]^\top \in \mathbb{R}^2$ is the state vector, $u\in \mathbb{R}$ is the control action, and the parameters are $\delta = 0.2$, $\alpha = -1$, and $\beta = 1$. For this set of parameters, the Duffing oscillator exhibits highly nonlinear restoring forces with multiple equilibrium basins. Notably, it can be shown that for the autonomous case $u=0$, $x = [0,0]^\top$ is an unstable equilibrium point, whereas $x = [\pm 1,0]^\top$ are stable equilibrium points. The unstable equilibrium point admits a separatrix (stable manifold) that divides the state space into two basins of attraction for the two stable equilibria. Consequently, long-horizon predictions starting near the separatrix are particularly challenging. See Fig. \ref{fig:phaseportraitduffing} for an illustration.
\subsubsection{Open-loop performance}
To obtain the dataset for training each predictor, $M_H=2\times10^3$ long trajectories of the oscillator were simulated, each with 50 recorded time steps of length $T_S = 0.025$. In the reported comparisons, we use the first $H=20$ steps as the prediction horizon. We used an observable map $\Psi$ consisting of the (scaled) Legendre polynomials of total degree $d \leq 14$, which results in a lifted space dimension of $N = 120$. The initial states $x_0$ were uniformly sampled in the box $[-2,2]^2 \subset \mathbb{R}^2$, and the oscillator was simulated using the fourth-order Runge-Kutta method. The control trajectories were sampled using a pseudo-random binary sequence with values in $\pm U_\text{amp}=\pm 1$ to ensure sufficient excitation. The RLS baseline uses the same fixed dictionary with relaxation parameter $\gamma=0.005$ and a maximum of 200 iterations. To identify the condensed multi-step EDMD matrices, we used an elastic-net multi-step fit with ridge coefficient $10^{-3}$ and $\ell_1$ coefficient $10^{-2}$; for the pruned multi-step model, the threshold weight was set to $\epsilon_E = 10^{-1}$. The metric used to compare the different predictors is the mean-squared error (MSE), which is defined for various time steps $k$ as
\begingroup
\setlength{\eqntopsep}{4pt}
\setlength{\abovedisplayskip}{4pt}
\setlength{\belowdisplayskip}{4pt}
\setlength{\abovedisplayshortskip}{0pt}
\setlength{\belowdisplayshortskip}{3pt}
\begin{equation}\label{eq: mse}
\text{MSE}(k) = \frac{1}{M_H} \sum_{j=1}^{M_H} \;\big\lVert x_\text{true}(kT_s) - x_\text{pred} (kT_s)\big\rVert^2.
\end{equation}
\endgroup

Fig. \ref{fig:openloopduffing} shows the mean-squared error at each time step $k$ for various predictors. We observe the phenomenon described in Section \ref{sec: error}, whereby the single-step EDMD prediction error grows exponentially with the prediction horizon, whereas the predictions made by multi-step EDMD remain stable over the reported horizon. This phenomenon is again attributed to the exponential growth of the spectral radius of $\widehat{A}_{M}^k$ ($\rho(\widehat{A}_{M}) = 3.70$) for single-step EDMD, which amplifies prediction errors made at earlier time steps. The RLS baseline uses the same fixed dictionary without ridge regularization, with relaxation parameter $\gamma=0.005$ and 200 maximum iterations, and reduces the spectral radius to $\rho(\widehat{A}_{M}) = 1.40$, but its open-loop error remains much larger than that of the direct multi-step models. Furthermore, the variant with pruning yields similar MSE while reducing the dimensionality of the lift from $N=120$ to $N = 6$, which shows that redundant observables can be removed without significant loss of performance. 

\subsubsection{Closed-loop performance}
The conventional single-step Koopman model exhibited numerical instability during open-loop prediction for longer horizons. We therefore compare the single-step EDMD-MPC, RLS-MPC, multi-step EDMD-MPC, and pruned multi-step EDMD-MPC controllers under the same prediction horizon $H=20$. At each time step, we solve the Condensed Koopman-MPC problem \eqref{eq: condensed MPC} with the settings $u_{\min}=-1$, $u_{\max}=1$, state and terminal weights $Q=P=I_2$, input weight $R=0.01 I_m$, and no state constraints. Fig. \ref{fig:closedloopduffing} shows the state and control trajectories for all models. The controllers based on the multi-step models gradually stabilize the oscillator while satisfying the input constraints. In contrast, the single-step EDMD-based controller exhibits frequent excitations and transitions between the two stable equilibrium points, and the RLS controller improves upon single-step EDMD but still does not stabilize the system as reliably as the direct multi-step controllers. This behavior is consistent with the poor long-horizon predictions of the single-step and RLS predictors, particularly near the separatrix, where small perturbations in the initial conditions lead to large deviations in long-horizon behavior. The phase portraits for all closed-loop simulations are shown in Fig. \ref{fig:phaseportraitduffing}.

\begin{figure}[t]
  \centering
  \includegraphics[width=0.95\linewidth]{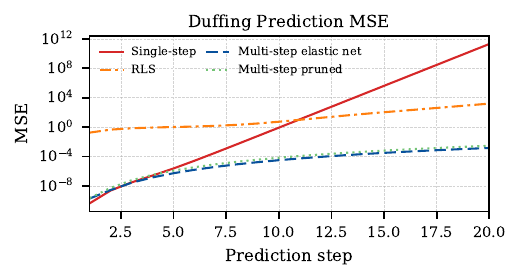}

  \vspace{-0.5cm}
  
  \caption{Prediction MSE comparison for the Duffing oscillator across different models, including the RLS baseline.}
  \label{fig:openloopduffing}
\end{figure}

\begin{figure}[t]
  \centering
  \includegraphics[width=0.95\linewidth]{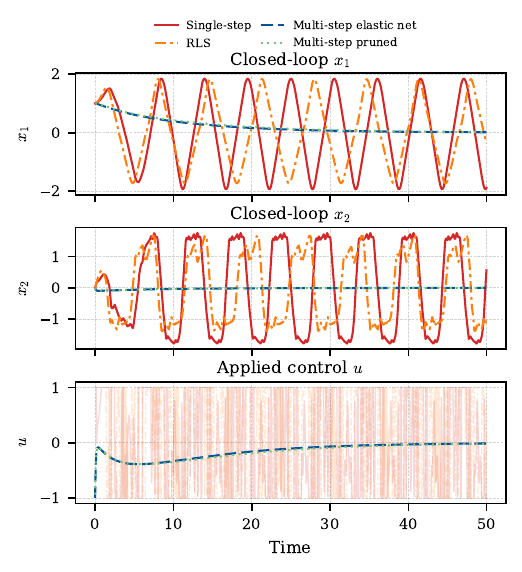}
  
  \vspace{-0.5cm}
  
  \caption{Closed-loop state (top) and control (bottom) trajectories for the oscillator across the single-step, RLS, multi-step, and pruned multi-step MPC controllers. The control trajectories from single-step and RLS are shown with increased transparency to improve clarity.}
  \label{fig:closedloopduffing}
\end{figure}

\begin{figure}[t]
  \centering
  \includegraphics[width=0.95\linewidth]{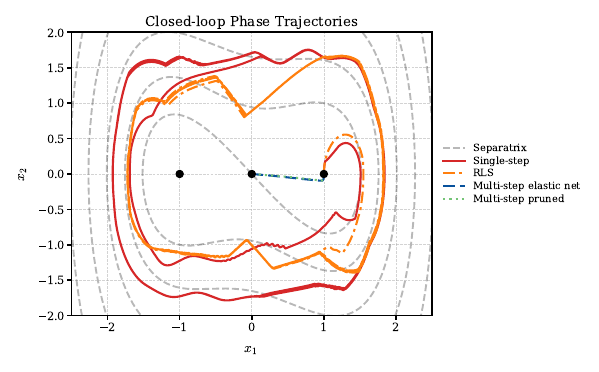}
  
  \vspace{-0.5cm}
  
  \caption{Closed-loop phase trajectories of the Duffing oscillator with separatrix. The separatrix separates the two basins of attraction of the equilibria (black dots). Multi-step MPC and its pruned variant converge near the origin, while single-step MPC and RLS-MPC exhibit poorer closed-loop behavior.}
  \label{fig:phaseportraitduffing}
\end{figure}

\subsection{Time-Varying Process Model}
We consider the following time-varying input-output process model \cite{mathworks_timevarying_mpc_plant}
\begin{equation}\label{eq:time_varying_plant}
y^{(3)} + 3\ddot y + 2\dot y + \left(6+\sin(5t)\right)y
= 5\dot v + \left(5+2\cos(2.5t)\right)v .
\end{equation}
Defining the state and applied MPC input as
$x := [y,\dot y,\ddot y,v]^\top$ and $u := \dot v$, respectively, the system can be written as
\begin{equation}
\begin{aligned}
\dot{x}_1 &= x_2,\qquad \dot{x}_2 = x_3,\\
\dot{x}_3 &= -3x_3 -2x_2 -\left(6+\sin(5t)\right)x_1 \\
           &\quad + \left(5+2\cos(2.5t)\right)x_4 + 5u,\\
\dot{x}_4 &= u.
\end{aligned}
\end{equation}
This example provides a useful test case for the proposed framework because
the time variation is periodic and can be encoded through phase-dependent
observables.
\subsubsection{Open-loop performance}
As in the previous case study, we simulated $M_H=900$ long trajectories of the time-varying system to obtain training data for the predictors, each with 100 recorded time steps of length $T_s=0.02$. We used a phase-dependent observable map based on the phase variable $\theta=2.5t$. Specifically, the observable map contains Fourier features up to the fourth harmonic of $\theta$ and their products with each state coordinate. We let
\begin{equation}
\eta(\theta) :=
[\sin\theta,\cos\theta,\ldots,\sin(4\theta),\cos(4\theta)]^\top .
\end{equation}
The resulting observable map is
\begin{equation}
\begin{aligned}
\Psi(x,\theta) :=
[&1,\;x^\top,\;\eta(\theta)^\top,\;x_1\eta(\theta)^\top,\;x_2\eta(\theta)^\top,\\
 &x_3\eta(\theta)^\top,\;x_4\eta(\theta)^\top]^\top,
\end{aligned}
\end{equation}
which results in a lifted dimension of $N=45$. Because the system \eqref{eq:time_varying_plant} contains the term $2\cos(\theta) u$, if we only use $v = \dot u$ as control, we will only learn $\theta$-averaged coefficients with respect to $v$. Hence, in both training and closed-loop MPC, we use lifted controls
\[
\mathbf v = [v, \; \cos(\theta) v, \; \sin(\theta) v]\,,
\]
which significantly improves open-loop prediction compared to just $v$. Note that at any time $t$ (which is known), $\cos(\theta), \sin(\theta)$ are deterministic. Therefore, despite control lifting, the predicted state $\hat x_k$ is still linear with respect to $v_0, \cdots, v_{k-1}$ in all models we consider, enabling QP-based MPC. The system \eqref{eq:time_varying_plant} was simulated using the fourth-order Runge-Kutta method, and the control trajectories were sampled as random piecewise-constant signals with amplitudes uniformly distributed in $[-1.5,1.5]$ to ensure sufficient excitation. To identify the EDMD matrices, the single-step EDMD baseline uses a ridge regularization coefficient of $\beta = 5\times10^{-3}$. The RLS baseline uses the same fixed dictionary with relaxation parameter $\gamma=0.005$, $20$ relaxed iterations, and ridge regularization coefficient $\beta = 5\times10^{-3}$ in each shifted refit. To identify the condensed multi-step EDMD matrices, we used an elastic-net multi-step fit with ridge coefficient $\beta = 10^{-3}$ and $\ell_1$ coefficient $\tau = 10^{-2}$; for the pruned multi-step model, the threshold weight was set to $\epsilon_E = 10^{-1}$, which reduces the observable dimension from $N=45$ to $N=15$.

Fig. \ref{fig:openlooptimevarying} shows the mean-squared error at each time step $k$ for various predictors. For the identified single-step lifted matrices, the spectral norms (largest singular values) are $\|\widehat{A}_{M}\|_2=1.02883 > 1$ for single-step EDMD and $\|\widehat{A}_{M}\|_2=1.00005$ for RLS. For single-step EDMD, we observe the same phenomenon whereby the prediction error grows exponentially with the prediction horizon; this effect is mitigated in RLS because of its lower spectral norm. The predictions made by both multi-step EDMD models remain stable over the reported horizon, and pruning the observables does not materially affect the prediction error.

\begin{figure}[t]
  \centering
  \includegraphics[width=0.95\linewidth]{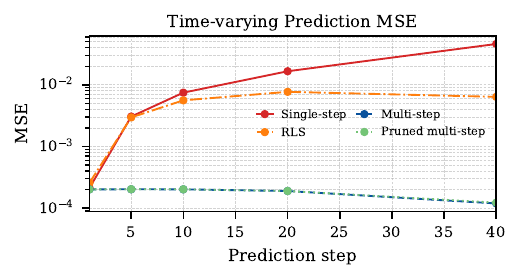}

  \vspace{-0.5cm}

  \caption{Prediction MSE comparison for the time-varying output-delay system across different models, including the RLS baseline.}
  \label{fig:openlooptimevarying}
\end{figure}

\begin{figure}[t]
  \centering
  \includegraphics[width=0.95\linewidth]{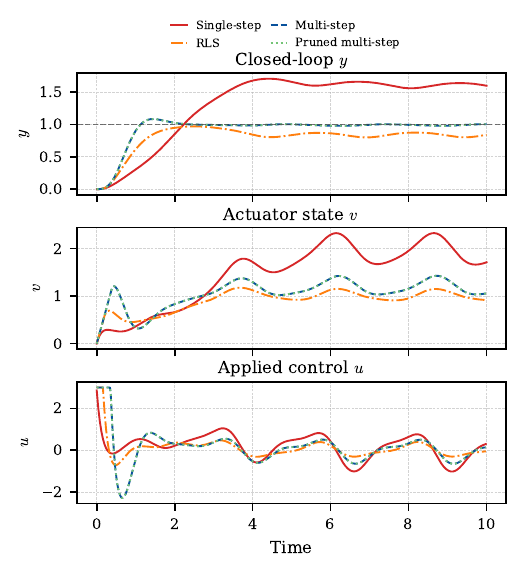}

  \vspace{-0.5cm}

  \caption{Closed-loop output, actuator state, and input-rate trajectories for the time-varying output-delay system across the single-step, RLS, multi-step, and pruned multi-step MPC controllers.}
  \label{fig:closedlooptimevarying}
\end{figure}
\subsubsection{Closed-loop performance}
Similar to the Duffing case study, we compare the single-step EDMD-MPC, RLS-MPC, multi-step EDMD-MPC, and pruned multi-step EDMD-MPC controllers under the same prediction horizon $H=40$. At each time step, we solve the Condensed Koopman-MPC problem \eqref{eq: condensed MPC} with target state $[1,0,0,0]^\top$, stage-state weight $Q=\mathrm{diag}(1,0.04,0.02,0)$, scalar control-action penalty $R=0.01$ on $u$, an additional quadratic move penalty $0.002$ on successive planned control differences $\sum_{j=0}^{H-2}(u_{k+j+1}-u_{k+j})^2$ for control trajectory smoothing, box constraint $-3\leq u \leq 3$, and no additional state constraints. In the MPC implementation, only the first optimized control action is applied at each update, and it is held constant over the sampling interval $T_S=0.02$. Figure \ref{fig:closedlooptimevarying} shows the closed-loop rollout trajectories for all controllers starting from the initial state $x=0$, including the regulated output $y$, the actuator state $v$, and the optimized control action $u$. The corresponding integral absolute errors for single-step EDMD-MPC, RLS-MPC, multi-step EDMD-MPC, and pruned multi-step EDMD-MPC are $5.84$, $2.06$, $0.825$, and $0.825$, respectively. The controllers based on the multi-step models stabilize the system at the desired point. In contrast, the controllers based on the single-step models exhibit a noticeable offset, which is likely attributable to poorer long-horizon predictions.

Overall, the open-loop and closed-loop simulations demonstrate that the multi-step Koopman operators provide greater predictive accuracy and more stable MPC operation than the single-step and RLS predictors in this strongly nonlinear dynamical regime. The pruning process also enables the removal of redundant observables in the multi-step predictor and does not significantly affect the predictive accuracy or closed-loop performance of our multi-step Koopman predictors.

\subsection{Discussion}
In both open-loop and closed-loop simulations, the relative advantage of multi-step EDMD is much more pronounced for the Duffing oscillator (Section 6.1) than for the time-varying system (Section 6.2). This suggests that multi-step EDMD provides a more pronounced advantage in dynamical regimes with stronger sensitivity to prediction error, such as highly nonlinear near-separatrix behavior.

We also performed an ablation study on the time-varying output-delay system in which all identification regularization was removed by setting the one-step ridge penalty, the RLS ridge penalty, and the multi-step elastic-net penalties to zero. In this setting, multi-step EDMD only outperforms single-step EDMD marginally in open-loop simulations, and the four controllers become nearly indistinguishable in the closed-loop simulations. Single-step EDMD without ridge regularization performs well because, with $T_s$ so small, the system \eqref{eq:time_varying_plant} is well approximated by the linearization through our comprehensive observable dictionary. This is not a contradiction of the error-propagation thesis; our claim is that recursively propagated one-step models can amplify identification error through repeated application of the learned map if prediction error exists, not that they must do so under every regularization setting or in every benign regime. As further evidence of our thesis, when the ridge penalty (which is commonly used for matrix conditioning) is increased from $10^{-5}$ to $8\times 10^{-3}$, single-step performance in both open- and closed-loop simulations degrades dramatically, while multi-step closed-loop performance changes very little, as evidenced by Fig. \ref{fig:openlooptimevarying}, Fig. \ref{fig:closedlooptimevarying}, and the ridge-sweep results in our GitHub repository. With highly correlated observables, ridge regression introduces errors in the identified matrices. Multi-step EDMD is much more robust to this source of error because it avoids error accumulation over the prediction horizon by design.

Although these numerical simulations demonstrate empirically that improved long-horizon prediction can benefit closed-loop performance, prediction accuracy alone does not always guarantee improved closed-loop performance. While this work focuses on convex least-squares identification for causal long-horizon prediction, an important future direction is to investigate whether techniques from \cite{de2024koopman} can be adapted to establish recursive feasibility and robust closed-loop stability guarantees for the proposed multi-step Koopman-MPC framework. Although the known periodic time dependence in Section 6.2 is represented through phase-dependent Fourier observables, the finite-horizon matrices $\mathbf E$ and $\mathbf F$ remain fixed with respect to time, after offline identification. Related iterative Koopman learning methods for unknown time-varying systems \cite{zhang2026deep} suggest a promising direction for adapting these matrices online, although such extensions might require larger sample-efficiency and sufficient regularization, due to the larger number of parameters.

Finally, we also considered an Allen--Cahn PDE case study. Due to spatial limitations, we do not include it in the main manuscript, but the corresponding results and code are provided in the GitHub repository. This example serves as an additional large-scale benchmark illustrating how the proposed multi-step framework extends beyond low-dimensional ODE systems to spatially distributed dynamics.

\section{Conclusion}
This paper proposes a multi-step EDMD framework for Koopman-based MPC that directly learns the condensed multi-step state-control mapping required for QP-based MPC formulations. The proposed approach preserves a convex least-squares structure while aligning the learning objective with the finite-horizon nature of MPC. We further showed that the identification problem decomposes across prediction horizons and state coordinates, enabling parallelized computation and row-wise $\ell_1$-regularization for automatic dictionary pruning. A non-asymptotic error analysis provides a rigorous explanation for the observed improved performance of the proposed multi-step EDMD for long-horizon predictions. Numerical simulations demonstrate that our framework achieves low multi-step prediction errors and effective closed-loop state regulation compared to conventional EDMD. Future work aims to analyze stochastic extensions of the multi-step Koopman operator \cite{xu2025data,vcrnjaric2020koopman} and integrate the framework with stochastic MPC (SMPC) formulations \cite{mesbah2016stochastic,heirung2018stochastic} and the robust Koopman-MPC framework \cite{zhang2022robust}, which demonstrates one representative approach for incorporating the approximation error bounds into a robust Koopman-based tube MPC framework.

\begin{ack}                              
Wallace Tan is supported by the MathWorks Fellowship. Liang Wu and Richard Braatz were partially supported by the U.S. Food and Drug Administration under the FDA BAA-22-00123 program, Award Number 75F40122C00200. This research was also supported by the Ralph O'Connor Sustainable Energy Institute at Johns Hopkins University, and by the U.S. DOE Office of Science, ASCR program under the Scientific Discovery through Advanced Computing (SciDAC) Institute ``LEADS: LEarning-Accelerated Domain Science.''
\end{ack}

\bibliographystyle{plain}
\bibliography{ref}           

\appendix
\newcommand{\appendixdisplayspacing}{%
    \setlength{\eqntopsep}{6pt plus 1pt minus 1pt}%
    \setlength{\abovedisplayskip}{6pt plus 1pt minus 1pt}%
    \setlength{\belowdisplayskip}{6pt plus 1pt minus 1pt}%
    \setlength{\abovedisplayshortskip}{0pt plus 1pt}%
    \setlength{\belowdisplayshortskip}{3pt plus 1pt minus 1pt}%
}
\appendixdisplayspacing
\AddToHook{cmd/normalsize/after}{\appendixdisplayspacing}
\AddToHook{cmd/small/after}{\appendixdisplayspacing}
\section{Proof of Lemma \ref{lem:regression}}\label{app:regression}

\textbf{Proof of Lemma \ref{lem:regression}}:
Let $\{(\chi_j,Y_j)\}_{j=1}^M$ be i.i.d.\ samples of $(\chi,Y)$, where
$\chi\in\mathbb R^{N_\chi}$ and $Y\in\mathbb R$. By orthonormality of the input
features, $\mathbb E[\chi_j\chi_j^\top]=I_{N_\chi}$. Moreover, the assumed
uniform norm bound implies $\|\chi_j\|_2\le B_\chi$ almost surely. Define
$Z_j:=\chi_j\chi_j^\top-I_{N_\chi}$. Then $\mathbb E[Z_j]=0$ and
$\|Z_j\|\le \|\chi_j\|_2^2+1\le B_\chi^2+1$. Furthermore,
\begin{equation}
\mathbb E[Z_j^2]
=
\mathbb E\!\left[\|\chi_j\|_2^2\chi_j\chi_j^\top\right]
-
I_{N_\chi}.
\end{equation}
For any unit vector $u\in\mathbb R^{N_\chi}$,
\begin{equation}
u^\top\mathbb E[Z_j^2]u
\le
B_\chi^2\,\mathbb E[(u^\top\chi_j)^2]-1
=
B_\chi^2-1.
\end{equation}
Hence $\|\mathbb E[Z_j^2]\|\le B_\chi^2-1$, and therefore
$\big\|\sum_{j=1}^M \mathbb E[Z_j^2]\big\|\le M(B_\chi^2-1)$.

By the Matrix Bernstein Inequality, there exists a universal constant $\tilde C$
such that, with probability at least $1-\delta$,
\begin{align*}
\|\widehat\Sigma-I_{N_\chi}\|
&\le
\tilde C\!\left(
\frac{B_\chi^2 + 1}{M}\log\frac{2N_\chi}{\delta}
+
\sqrt{
\frac{B_\chi^2 - 1}{M}\log\frac{2N_\chi}{\delta}}
\right) \\
&=: \eta .
\end{align*}
For $\eta<1$, this inequality implies that
$\|\widehat\Sigma^{-1}-I_{N_\chi}\|\le \eta/(1-\eta)$. Next, define $g:=\mathbb E[\chi_jY_j]$. Since
$\mathbb E[\chi_j\chi_j^\top]=I_{N_\chi}$, the normal equations imply
$\beta_\star=g$. For any unit vector $u\in\mathbb R^{N_\chi}$,
\begin{equation}
|u^\top g|
\le
\sqrt{\mathbb E[Y_j^2]\mathbb E[(u^\top\chi_j)^2]}
\le
M_y,
\end{equation}
so $\|g\|_2\le M_y$. Define $v_j:=\chi_jY_j-g$ and
$V_j:=\begin{bmatrix}0 & v_j^\top \\ v_j & 0\end{bmatrix}$. Then
$\mathbb E[V_j]=0$. Moreover, $\|V_j\|=\|v_j\|_2$. Since $\|g\|_2\le M_y$ and, by assumption, $|Y_j|\le M_y$ almost surely,
$
\|v_j\|_2
\le
\|\chi_jY_j\|_2+\|g\|_2
\le
M_yB_\chi+M_y.
$
Also, since
$
N_\chi
=
\mathbb E[\|\chi_j\|_2^2]
\le
B_\chi^2,
$
we have $B_\chi\ge 1$. Hence
$
\|V_j\|
=
\|v_j\|_2
\le
2M_yB_\chi.
$
Furthermore,
\[
V_j^2
=
\begin{bmatrix}
\|v_j\|_2^2 & 0\\
0 & v_jv_j^\top
\end{bmatrix},
\]
and therefore
$
\|V_j^2\|
=
\|v_j\|_2^2
\le
4M_y^2B_\chi^2.
$
Thus
$$
\left\|
\sum_{j=1}^M\mathbb E[V_j^2]
\right\|
\le
4M_y^2MB_\chi^2.
$$
Applying the Matrix Bernstein Inequality to the self-adjoint matrices $V_j$
gives, with probability at least $1-\delta$,
\begin{align*}
\|\widehat g-g\|_2
&\le
2\tilde C\!\bigg(
\frac{M_yB_\chi}{M}\log\frac{2(N_\chi+1)}{\delta} \\
&\quad+
\sqrt{
\frac{M_y^2B_\chi^2}{M}
\log\frac{2(N_\chi+1)}{\delta}}
\bigg).
\end{align*}
Here we used that
\[
\left\|
\frac1M\sum_{j=1}^M V_j
\right\|
=
\left\|
\begin{bmatrix}
0 & (\widehat g-g)^\top\\
\widehat g-g & 0
\end{bmatrix}
\right\|
=
\|\widehat g-g\|_2.
\]

Set $\delta=M^{-2}$. Assume that
$
M\ge C_0 B_\chi^2\log M 
$
for a sufficiently large universal constant $C_0$. Since
$
N_\chi
=
\mathbb E[\|\chi_j\|_2^2]
\le
B_\chi^2,
$
we have 
$
L_M:=\log\!\bigl(2(N_\chi+1)M^2\bigr)
\le
\log\!\bigl(4B_\chi^2M^2\bigr).
$
Using $B_\chi^2\le M/(C_0\log M)$, this gives
$$
L_M
\le
\log\!\left(\frac{4M^3}{C_0\log M}\right)
\lesssim
\log M .
$$
Therefore, with probability at least $1-2M^{-2}$,
\[
\eta
\le
C_1
\left(
\frac{(B_\chi^2+1)\log M}{M}
+
B_\chi\sqrt{\frac{\log M}{M}}
\right)
\le
C_1
B_\chi\sqrt{\frac{\log M}{M}} .
\]
for some universal constant $C_1 > 0$. Indeed, since $N_\chi\le B_\chi^2$, we have $B_\chi\ge 1$, and the assumption
$M\ge C_0B_\chi^2\log M$ implies
\[
\frac{(B_\chi^2+1)\log M}{M}
\lesssim
\frac{B_\chi^2\log M}{M}
\le
B_\chi\sqrt{\frac{\log M}{M}} .
\]
Increasing $C_0$ if necessary, the same assumption implies $\eta\le 1/2$.
Similarly,
\begin{align*}
\|\widehat g-g\|_2
& \le
C_2
M_y
\left(
\frac{B_\chi\log M}{M}
+
B_\chi\sqrt{\frac{\log M}{M}}
\right) \\
& \le
C_2
M_yB_\chi\sqrt{\frac{\log M}{M}} .
\end{align*}
for some universal constant $C_2 > 0$. Hence
\begin{equation}
\begin{aligned}
\|\widehat\beta-\beta_\star\|_2
&=
\|\widehat\Sigma^{-1}\widehat g-g\|_2 \\
&\le
\|\widehat\Sigma^{-1}-I_{N_\chi}\|\,\|g\|_2
+
\|\widehat\Sigma^{-1}\|\,\|\widehat g-g\|_2 \\
&\le
\frac{\eta}{1-\eta}M_y
+
\Bigl(1+\frac{\eta}{1-\eta}\Bigr)\|\widehat g-g\|_2 \\
&\le
2\eta M_y+2\|\widehat g-g\|_2 \le
\tilde C_l
M_yB_\chi
\sqrt{\frac{\log M}{M}},
\end{aligned}
\end{equation}
for some universal constant $\tilde C_l>0$.

Finally, let
$
\Delta:=\widehat\beta-\beta_\star.
$
By the population normal equations,
\[
\mathbb E[\chi(Y-\beta_\star^\top\chi)]
=
g-\Sigma\beta_\star
=
0.
\]
Therefore,
\begin{equation}
\begin{aligned}
&
\mathbb E\!\left[
\left|Y-\widehat\beta^\top\chi\right|^2
\right]
-
\mathbb E\!\left[
\left|Y-\beta_\star^\top\chi\right|^2
\right] \\
&\quad
=
\mathbb E\!\left[
\left|
Y-\beta_\star^\top\chi-\Delta^\top\chi
\right|^2
\right]
-
\mathbb E\!\left[
\left|Y-\beta_\star^\top\chi\right|^2
\right] \\
&\quad
=
-2\Delta^\top
\mathbb E[\chi(Y-\beta_\star^\top\chi)]
+
\Delta^\top\Sigma\Delta =
\Delta^\top\Sigma\Delta.
\end{aligned}
\end{equation}
Since $\Sigma=I_{N_\chi}$, this gives
\begin{equation}
\begin{aligned}
&
\mathbb E\!\left[
\left|Y-\widehat\beta^\top\chi\right|^2
\right]
-
\mathbb E\!\left[
\left|Y-\beta_\star^\top\chi\right|^2
\right] \\
&\hspace{2em}
=
\|\widehat\beta-\beta_\star\|_2^2 \le
C_l\,
\frac{M_y^2B_\chi^2\log M}{M},
\end{aligned}
\end{equation}
for some universal constant $C_l>0$.

\section{Proof of Proposition \ref{prop:ss-edmd-error}}\label{app:singlestep}
For $i=1,\ldots,N$, let $[\Kop_u\Psi]_i$ denote the $i$-th entry of
$\Kop_u\Psi$, and let $\widehat A_{M,i}$ and $\widehat B_{M,i}$ denote the
$i$-th rows of $\widehat A_M$ and $\widehat B_M$, respectively. Define the
feature vector
$
    \chi(x,u):=
    [\Psi(x)^\top \; u^\top]^\top .
$
By Assumption~\ref{assum:ss-features}, $\mathbb E[\chi\chi^\top]=I_{N+n_u}$.
Moreover,
\[
    \|\chi(x,u)\|_2^2
    =
    \|\Psi(x)\|_2^2+\|u\|_2^2
    \le
    NB_\psi^2+R_u^2,
\]
and, since each observable is uniformly bounded,
$
    |[\Kop_u\Psi(x)]_i|\le B_\psi .
$ By the row-wise decomposition of the least-squares loss, the estimator
$(\widehat A_M,\widehat B_M)$ can be obtained by solving $N$ independent scalar
regression problems. In particular, for each $i=1,\ldots,N$,
$
    \widehat\beta_i
    :=
    [\widehat A_{M,i}\;\; \widehat B_{M,i}]^\top
$
satisfies
\[
    \widehat\beta_i
    \in
    \operatorname*{arg\,min}_{\beta\in\mathbb R^{N+n_u}}
    \frac{1}{M}\sum_{j=1}^M
    \left|
        [\Kop_{u_j}\Psi(x_j)]_i
        -
        \beta^\top \chi(x_j,u_j)
    \right|^2 .
\]
Similarly, if $A_{*,i}$ and $B_{*,i}$ denote the $i$-th rows of $A_*$ and
$B_*$, then
$
    \beta_{*,i}:=[A_{*,i}\;\;B_{*,i}]^\top
$
is a population minimizer for the corresponding scalar regression problem:
\begin{equation}\label{eq:ss-scalar-pop-regression}
    \beta_{*,i}
    \in
    \operatorname*{arg\,min}_{\beta\in\mathbb R^{N+n_u}}
        \mathbb E_{\mu_x\otimes\mu_u}
    \left[
        \left|
            Y_i(x,u)-\beta^\top\chi(x,u)
        \right|^2
    \right].
\end{equation}
Applying Lemma~\ref{lem:regression} with
$
    Y=[\Kop_u\Psi(x)]_i,
    B_\chi^2=NB_\psi^2+R_u^2,
    M_y=B_\psi,
$ gives, for each fixed $i=1,\ldots,N$, with probability at least $1-2M^{-2}$, \begin{small}
\vspace{-25.0pt}
\begin{align*}
    &\mathbb E
    \left[
        \left|
            [\Kop_u\Psi(x)]_i
            -
            \widehat\beta_i^\top\chi(x,u)
        \right|^2
    \right] \\
    &\qquad\le
    \mathbb E
    \left[
        \left|
            [\Kop_u\Psi(x)]_i
            -
            \beta_{*,i}^\top\chi(x,u)
        \right|^2
    \right]
    +
    C_l
    \frac{
        B_\psi^2\left(NB_\psi^2+R_u^2\right)\log M
    }{M}.
\end{align*}
\end{small}
Taking a union bound over all $i=1,\ldots,N$, the above inequalities hold
simultaneously with probability at least $1-2NM^{-2}$. Summing over all
coordinates gives \eqref{eqn:ssEDMDerror}.
\end{document}